\documentclass[letterpaper, 10 pt]{IEEEtran}

\usepackage{cite}
\usepackage{amsmath,amssymb,amsfonts,amsthm}
\usepackage{hyperref}
\usepackage{epstopdf,color}
\usepackage{textcomp}
\usepackage{graphicx,color}
\graphicspath{{figures/}}
\usepackage{mathrsfs}
\usepackage{subfigure}
\usepackage{url}
\usepackage{color}
\usepackage{dsfont}
\usepackage{bbm}
\usepackage{booktabs}
\usepackage{array}
\usepackage[table]{xcolor}
\usepackage{yfonts}
\usepackage{cleveref} 
\usepackage{tikz}
\usepackage{pgfplots}
\usepackage{multirow}
\usepackage{textcomp}
\usepackage{tabularx}
\usepackage{placeins}
\usepackage{float}
\usepackage{nccmath}
\usepackage{accents}
\usepackage{multicol}
\usepackage{etoolbox, refcount}
\usepackage{chngcntr}
\usepackage{apptools}
\usepackage{relsize}
\usepackage[font=footnotesize]{caption}
\usepackage{algorithm}
\usepackage{algpseudocode}
\usepackage{bigints}

\newtheorem{theorem}{Theorem}[section]
\newtheorem{lemma}[theorem]{Lemma}

\newtheorem{corollary}[theorem]{Corollary}
\newtheorem{proposition}[theorem]{Proposition}
\newtheorem{problem}{Problem}
\newtheorem{remark}[theorem]{Remark}
\newtheorem{assumption}{Assumption}

\DeclareMathAlphabet{\mymathbb}{U}{BOONDOX-ds}{m}{n}

\newcommand{\setdef}[2]{\{#1 \; : \; #2\}}

\newcommand{\Ac}{\mathcal{A}}
\newcommand{\Rc}{\mathcal{R}}
\newcommand{\Ic}{\mathcal{I}}
\newcommand{\Jc}{\mathcal{J}}
\newcommand{\Sc}{\mathcal{S}}

\newcommand{\Var}{\mathrm{Var}}

\newcommand{\VaR}{\operatorname{Var}}

\newcommand{\real}{\mathbb{R}}

\newcommand{\Cc}{\mathcal{C}}

\newcommand{\Nc}{\mathcal{N}}

\newcommand{\argmin}[2] {\mathrm{arg}\min_{#1}#2}

\DeclareSymbolFont{bbold}{U}{bbold}{m}{n}
\DeclareSymbolFontAlphabet{\mathbbold}{bbold}

\newcommand{\norm}[1]{\lVert#1\rVert}

\newcommand{\Ite}{\operatorname{It}_{\epsilon}}

\newcommand\oprocendsymbol{\hbox{$\bullet$}}
\newcommand\oprocend{\relax\ifmmode\else\unskip\hfill\fi\oprocendsymbol}

\newcommand*{\QEDA}{\hfill\ensuremath{\blacksquare}}%

\newcommand\xqed[1]{%
  \leavevmode\unskip\penalty9999 \hbox{}\nobreak\hfill
  \quad\hbox{#1}}
\newcommand\demo{\xqed{$\bullet$}}

\newcounter{countitems}
\newcounter{nextitemizecount}
\newcommand{\setupcountitems}{%
  \stepcounter{nextitemizecount}%
  \setcounter{countitems}{0}%
  \preto\item{\stepcounter{countitems}}%
}
\makeatletter
\newcommand{\computecountitems}{%
  \edef\@currentlabel{\number\c@countitems}%
  \label{countitems@\number\numexpr\value{nextitemizecount}-1\relax}%
}
\newcommand{\nextitemizecount}{%
  \getrefnumber{countitems@\number\c@nextitemizecount}%
}
\newcommand{\previtemizecount}{%
  \getrefnumber{countitems@\number\numexpr\value{nextitemizecount}-1\relax}%
}
\makeatother    
{\end{itemize}%
\unskip\computecountitems\ifnumcomp{\previtemizecount}{>}{4}{\end{multicols}}{}}

\newcommand{\longthmtitle}[1]{\mbox{}\emph{(#1):}}

\newcommand{\comment}[1]{} 
\newcolumntype{P}[1]{>{\centering\arraybackslash}p{#1}}
\allowdisplaybreaks

\begin{document}
\title{\Large \bf 
  Off-Policy Reinforcement Learning with Anytime Safety Guarantees via
  Robust Safe Gradient Flow\thanks{A preliminary version of this work
    appeared as \cite{PM-AM-JC:25-l4dc} at the 2025 Learning for
    Dynamics and Control Conference. This work was partially supported
    by ONR Award N00014-23-1-2353.}}













\author{Pol Mestres \qquad Arnau Marzabal \qquad Jorge Cort{\'e}s
  \thanks{The authors are with the Department of Mechanical and
    Aerospace Engineering, University of California, San Diego,
    \{pomestre,amarzabal,cortes\}@ucsd.edu}%
}

\maketitle

\begin{abstract}
  This paper considers the problem of solving constrained
  reinforcement learning (RL) problems with anytime guarantees,
  meaning that the algorithmic solution must yield a
  constraint-satisfying policy at every iteration of its evolution.
  Our design is based on a discretization of the Robust Safe Gradient
  Flow (RSGF), a continuous-time dynamics for anytime constrained
  optimization whose forward invariance and stability properties we
  formally characterize.  The proposed strategy, termed RSGF-RL, is an
  off-policy algorithm which uses episodic data to estimate the value
  functions and their gradients and updates the policy parameters by
  solving a convex quadratically constrained quadratic program.  Our
  technical analysis combines statistical analysis, the theory of
  stochastic approximation, and convex analysis to determine the
  number of episodes sufficient to ensure that safe policies are
  updated to safe policies and to recover from an unsafe policy, both
  with an arbitrary user-specified probability, and to establish the
  asymptotic convergence to the set of KKT points of the RL problem
  almost surely.  Simulations on a navigation example and the
  cart-pole system illustrate the superior performance of RSGF-RL with
  respect to the state of the art.
\end{abstract}

\section{Introduction}
Reinforcement learning (RL) seeks to find an optimal decision policy
by having an agent interact with its environment through trial and
error.  At any given state, an action taken by the agent makes them
transition to a new state with some probability, after which they
incur an associated reward.  The optimal policy is that which
maximizes a prespecified long-horizon cumulative reward.  Today,
RL-based methods are pervasive in a wide range of technological
applications of machine learning and artificial intelligence in
society.  However, the use of RL in safety-critical applications
(e.g., autonomous driving, healthcare, or energy management) requires
additional precautions, because the process of trial-and-error
can lead the agent towards unsafe regions,
with potentially catastrophic consequences.  This has sparked the
development of safe RL techniques that seek to find optimal policies
meeting desired safety specifications.  In this paper, we design an
algorithm to solve constrained RL problems in an anytime fashion,
meaning that the algorithm satisfies the constraints at every iterate.

\subsubsection*{Literature Review}

Safe RL has been actively pursued in recent years,
see~\cite{LB-MG-AWH-ZY-SZ-JP-APS:22,YL-AH-XL:21,JG-FF:15,SG-LY-YD-GC-FW-JW-AK:24}
for comprehensive surveys on the subject. Here, we discuss the works
best aligned with the approach to safe RL taken here. Safety
constraints in RL are often expressed as \textit{cumulative
  constraints}, which require the expected value of a sum of costs
over a given time horizon to be kept below a certain
threshold~\cite{EA:99,JA-DH-AT-PA:17,YC-MG-LJ-MP:17,SP-LC-MCF-AR:19}.
Markov Decision Processes with such type of constraints are referred
to as Constrained Markov Decision Processes (CMDPs).  A standard
approach to solve CMDPs are primal-dual
methods~\cite{SP-MCF-LFOC-AR:23,DD-KZ-TB-MJ:20}, which take a gradient
ascent step in the primal variable and a gradient descent step in the
dual variable.  For finite state and action MDPs with a special type
of transition functions,~\cite{DD-XW-ZY-ZW-MJ:21} shows that such
primal-dual scheme converges to the optimal policy.  Similarly, for
continuous state and action spaces,~\cite{SP-LC-MCF-AR:19}
also provides a primal-dual scheme that provably converges to the
optimal policy.  However, these guarantees require solving an
unconstrained RL algorithm at every iteration, which makes the
algorithm computationally hard to execute
%
%
(although practical implementations are given
in~\cite{SP-MCF-LFOC-AR:23}).  Furthermore, primal-dual schemes can
lead to safety violations during the training process, which
compromises their implementation in physical domains.
Although there exist implementations of primal-dual methods that
guarantee safety during training, these are either limited to
particular policy parametrizations~\cite{SZ-TTD-JR:22} or solve a
relaxed version of the problem and hence introduce an optimality
gap~\cite{QB-ASB-MA-AK-VA:22}.
Beyond primal-dual methods, there exist other
algorithms in the literature whose goal is to provide safety
guarantees during training.  For example,~\cite{JA-DH-AT-PA:17}
proposes CPO, an algorithm that is solely based on primal updates and
that enjoys safety guarantees at every iteration.  However, performing
the exact policy update is computationally intensive, and the proposed
practical implementations employs a first-order approximation of the
objective and constraints that might violate the safety constraints
during training.  On the other hand,~\cite{YL-JD-XL:20} introduces
IPO, another primal method that includes the safety constraints as
penalty terms in the objective function, and also guarantees the
satisfaction of the safety constraints during training. However, this
algorithm presupposes the existence of a safe initial policy and its
convergence properties are not studied.  The method proposed
in~\cite{YC-ON-ED-MG:18} leverages Lyapunov functions to guarantee the
satisfaction of constraints during training.  However, the method
proposed to search for such Lyapunov functions might be
computationally intensive, and it is only shown to converge for a
limited class of problems.  On the other hand,~\cite{JM-XZ:24}
introduces an algorithm for finite state and action CMDPs that
guarantees that trajectories satisfy budget constraints at all times.
It is also possible to optimize over a class of truncated policies so
that unsafe actions have probability zero, as
in~\cite{WS-VKS-KCK-SS-JL-VG-BMS:24}, but such restrictions also
introduce an optimality gap, which is not formally quantified.

The methods described above are all on-policy, i.e., they rely on
trajectories from the current policy iterate to generate the estimates
needed for the algorithm execution.  Instead, here we pursue the
design of off-policy methods, where trajectories from other policies
can be used to generate the estimates.  Such methods enable the use of
datasets of trajectories obtained offline or in previous iterations,
significantly enhancing the efficiency of the algorithm implementation.


\subsubsection*{Statement of Contributions}

The paper contributions are:
\begin{enumerate}
\item we introduce a continuous-time algorithm for anytime constrained
  optimization termed Robust Safe Gradient Flow (RSGF).  We identify
  conditions under which the RSGF is well defined and locally
  Lipschitz. We also establish the equivalence between its equilibria
  and the KKT points of the constrained optimization problem, and show
  forward invariance of the constraint set and convergence to the set
  of KKT points;
\item we define estimates for the value functions defining the
  constrained RL problem as well as their gradients. These estimates
  are off-policy, i.e. the estimates of any given policy
  can be constructed using trajectories generated by other policies.
  We establish a range of statistical properties of these estimates,
  including their mean and bounds on their variance and tail probabilities;
\item we combine (i) and (ii) to introduce the off-policy Robust Safe
  Gradient Flow-based Reinforcement Learning (RSGF-RL). This algorithm
  is based on a discretization of RSGF and employs the off-policy
  estimates of the value function and their gradients.  By leveraging
  the statistical properties of the latter, we determine a sufficient
  number of episodes such that RSGF-RL updates safe (and unsafe but
  close to safe) policies to safe policies for any prescribed
  confidence.  Combining the properties of RSGF with the theory of
  stochastic approximation~\cite{HJK-DSC:78,VSB:08}, we also show that
  the iterates of RSGF-RL asymptotically converge to a KKT point almost surely, 
  and characterize its rate of convergence.
\item we illustrate the performance of RSGF-RL on a navigation example
  and the cart-pole system, and compare it against the state of the
  art.
\end{enumerate}
Preliminary results were presented in the conference
article~\cite{PM-AM-JC:25-l4dc}, whose focus was restricted to
on-policy data and a single safety constraint. Furthermore, the
convergence to the set of KKT points was only ensured in
expectation.  All of these are special cases of the present work.  The
generalization here from on-policy to off-policy data and the
establishment of almost sure convergence, along with the novel
technical treatment based on the dynamical properties of RSGF and the
theory of stochastic approximation, are instrumental in expanding the
applicability of the proposed framework.


\section{Preliminaries}

We introduce here the notation and basic notions on stability of
dynamical systems, Markov decision processes, and constraint
qualification in nonlinear programming.

\emph{Notation:} We denote by $\mathbb{Z}_{>0}$, $\real$, and
$\real_{\geq0}$ the set of positive integers, real, and nonnegative
real numbers, respectively. Given $N \in\mathbb{Z}_{>0}$, we let
$[N]=\{1,2,\hdots,N \}$.  For $N_1, N_2\in\mathbb{Z}$, we let
$[N_1:N_2] = \{ N_1+1, N_1+2, \hdots, N_2 \}$.  For $x\in\real^n$,
$\norm{x}$ denotes its Euclidean norm, and for $l\in[n]$, $x^{(l)}$ is
its $l$-th component.  Given a set $\Cc\subset\real^n$,
$\mathbbm{1}_{\Cc}$ is the indicator function of $\Cc$, which is such
that $\mathbbm{1}_{\Cc}(x) = 1$ if $x\in\Cc$ and
$\mathbbm{1}_{\Cc}(x) = 0$ otherwise.  We let $\mathbf{I}_n$ be the
$n$-dimensional identity matrix. 
Given a function $V:\real^n\to\real^m$,
we let
$\text{Im}(V) = \setdef{ V(\theta)\in\real^m}{ \theta\in\real^n }$
denote its image.  Given a random variable $X$ taking scalar values,
$X\sim \eta$ indicates $X$ is distributed according to a probability
distribution~$\eta$, $\mathbb{E}[X]$ denotes its expectation, and
$\VaR(X) = \mathbb{E}[(X-\mathbb{E}[X])^2]$ its variance.  Given
a set $S$, $P(S)$ denotes its power set, i.e., the collection of all
subsets of~$S$. The collection $\Sigma \subset P(S)$ is a
$\sigma$-algebra if and only if: (i) $S$ is in $\Sigma$, (ii) if
$A\in\Sigma$, the complement of $A$ is also in $\Sigma$, (iii) if
$\{ A_i \}_{i\in\mathbb{Z}>0}$ is a countable union of sets in
$\Sigma$, then $\bigcup_{i\in\mathbb{Z}_{>0}}A_i\in\Sigma$.

\subsubsection*{Stability of Dynamical Systems}
We recall here concepts on stability of dynamical systems
following~\cite{HK:02}.  Let $F:\real^n\to\real^n$ be a locally
Lipschitz vector field and consider the dynamical system
$\dot{z} = F(z)$. Local Lipschitzness ensures that for every initial
condition $x\in\real^n$ there exists $T>0$ and a unique trajectory
$z:[0,T]\to\real^n$ such that $z(0) = x$ and $\dot{z}(t) = F(z(t))$.
If the solution is defined for all $t\geq0$, then it is
\textit{forward complete}.  If every solution is forward complete, for
each $t \geq 0$, the \textit{flow map} is defined by the function
$\Phi_t:\real^n\to\real^n$ such that $\Phi_t(x) = z(t)$.  A set
$\mathcal{K}\subset\real^n$ is forward invariant
if, for every initial condition $x\in\mathcal{K}$, the trajectory with
initial condition at $x$ is forward complete and
$\Phi_{t}(x)\in\mathcal{K}$ for all $t\geq0$.

\subsubsection*{Constrained Markov Decision Processes}
Here we recall concepts on Constrained Markov Decision Processes
(CMDP) following~\cite{EA:99,RSS-AGB:18}.  A CMDP is given by a tuple
$ \mathcal{M} = (\Sc,\Ac,P,R_0,\{ R_j \}_{j=1}^q )$, with $q\in\mathbb{Z}_{>0}$. 
Here, $\Sc$ is a set of states,
$\Ac$ is a set of actions, and $P:\Sc\times\Ac\times\Sc\to[0,1]$ is a
probability transition function, where $P(s, a, s^{\prime})$
represents the probability that the agent transitions to state
$s^{\prime}\in\Sc$ given that it is at state $s\in\Sc$ and takes
action $a\in\Ac$.  Further, $R_0:\Sc\times\Ac\times\Sc\to\real$ and
$R_j:\Sc\times\Ac\times\Sc\to\real$ for $j\in[q]$ are functions: for every
$s\in\Sc$, $a\in\Ac$, and $s^{\prime}\in\Sc$, $R_0(s,a,s^{\prime})$ is
the reward associated with completing a task when an agent is at state
$s$, takes action $a$, and transitions to state $s^{\prime}$.
Instead, $R_j(s,a,s^{\prime})$ is the cost associated with a safety
constraint when an agent is at state $s$, takes action $a$, and
transitions to state $s^{\prime}$.  A policy $\pi$ for the CMDP is a
function that maps every state $s\in\Sc$ to a distribution over $\Ac$,
denoted as $\pi(\cdot|s)$:
here, $\pi(a|s)$ is the probability of taking action $a\in\Ac$ at
state $s\in\Sc$.

\emph{Constraint Qualifications in Nonlinear Programming:}
We summarize here various constraint qualification conditions from
nonlinear programming following~\cite{DPB:99,SB-LV:09,JL:95}. Let
$f, g_1 \hdots, g_q: \real^d\to\real$ be differentiable functions, and
consider a nonlinear optimization problem of the form
\begin{align}\label{eq:optimization-problem}
  \notag
  &\min\limits_{x\in\real^d} f(x)
  \\
  &\ \text{s.t.}\ g_j(x) \leq 0, \ j\in[q],
\end{align}
%
%
where $f, g_1, \hdots, g_q$ are continuously differentiable. Let the
active and inactive constraint sets be defined by
$ I_0(x) = \setdef{j\in[q]}{g_j(x) = 0}$ and
$ I_{-}(x) = \setdef{j\in[q]}{g_j(x) < 0}$. Then,
\begin{itemize}
\item Slater's condition (SC) holds
  for~\eqref{eq:optimization-problem} if there exists $x\in\real^d$
  such that $g_j(x) < 0$ for all $j\in[q]$;
\item the Mangasarian-Fromovitz Constraint Qualification (MFCQ) holds
  for~\eqref{eq:optimization-problem} at $x\in\real^d$ if there exist
  $\xi\in\real^d$ such that $\nabla g_j(x)^\top \xi < 0$ for all
  $j\in I_0(x)$;
\item the constant-rank condition (CRC) holds
  for~\eqref{eq:optimization-problem} at $x\in\real^d$ if there exists
  a neighborhood $\Nc$ of $x$ such that for all $I \! \subset \! I_0(x)$,
  $\text{rank}( \{ \nabla g_j(\bar{x}) \}_{j\in I})$ is constant for all
  $\bar{x}\in \Nc$.
\end{itemize}
If $x^*$ is a local minimizer of~\eqref{eq:optimization-problem}, and
MFCQ or CRC
hold at $x^*$, then there exist $u^*\in\real^q$ (which we refer to as
a Lagrange multiplier vector) such that the
\textit{Karush-Kuhn-Tucker} (KKT) conditions hold:
\begin{subequations}
  \begin{align}
    & \nabla f(x^*) + \sum_{i=1}^q u_j^* \nabla g_j(x^*) =
      0,
    \\
    &g_j(x^*) \leq 0, \ u_j^* \geq 0, \ u_j^* g_j(x^*) = 0 , \quad j \in [q] .
  \end{align}
  \label{eq:kkt-equations}
\end{subequations}
If~\eqref{eq:optimization-problem} is convex, $x^*$ is a local
minimizer, and SC holds, then the KKT
conditions~\eqref{eq:kkt-equations} also hold.  Any $x^*\in\real^d$
for which there exists $u^*\in\real^q$
satisfying~\eqref{eq:kkt-equations} is referred to as a KKT point
of~\eqref{eq:optimization-problem}.  We note that $u_j^*$ is the
Lagrange multiplier associated with the $j$-th constraint.

Given differentiable functions
$\tilde{f},\tilde{g}_1,\hdots,\tilde{g}_q:\real^d\times\real^c
\rightarrow \real$, consider the parametric nonlinear optimization
problem
\begin{align}\label{eq:parametric-optimization-problem}
  \notag
  &\min\limits_{x\in\real^d} \tilde{f}(x,z) \\
  &\ \text{s.t.} \ \tilde{g}_j(x,z) \leq 0, \ j\in[q].
\end{align}
Let 
$\tilde{I}_0(x,z) = \setdef{ j\in[q] }{ g_j(x,z) = 0 }$
be the set of active constraints.
We say that the constant-rank condition (CRC) holds
for~\eqref{eq:parametric-optimization-problem} at
$(x_0,z_0)\in\real^d\times\real^c$ if there exists a neighborhood
$\Nc$ of $(x_0,z_0)$ such that for any
$\tilde{I}\subset\tilde{I}_0(x_0,z_0)$ and $(x,z)\in\Nc$,
$\{ \nabla_x g_j(x,z) \}_{j\in\tilde{I}}$ has a constant rank.

\section{Problem Statement}\label{sec:problem-statement}
In this section we formalize the problem of solving constrained
reinforcement learning (RL) problems in an anytime fashion.  Given a
CMDP
$\mathcal{M} = (\mathcal{S},\mathcal{A},P,R_0, \{ R_j \}_{j=1}^q )$,
the goal is to maximize the cumulative reward while keeping the
cumulative costs below a threshold.  We consider a parametric
class of policies indexed by a vector $\theta\in\real^d$. We denote
the policy associated with $\theta$ as $\pi_{\theta}$.  Given a
distribution $\eta$ of initial states, a discount factor
$\gamma\in(0,1)$, and a time horizon $T\in\mathbb{Z}_{>0}$, we
consider the following problem:
\begin{align}\label{eq:safe-RL-problem-parameterized-policies}
  &\min_{\theta\in\real^d} \ V_0(\theta) \!=\! \mathbb{E}_{ \substack{a \sim
    \pi_{\theta}(\cdot|s) \\ \hspace{-0.7cm} s_0 \sim \eta} }
  \biggl[ \sum_{k=0}^{T} 
  -\gamma^k R_0(s_k, a_k, s_{k+1})
  \biggr] 
  \\  
  & \; \text{s.t.} \quad V_j(\theta) \!=\! \mathbb{E}_{ \substack{a \sim
    \pi_{\theta}(\cdot|s) \\ \hspace{-0.7cm} s_0 \sim \eta} }
  \biggl[ \sum_{k=0}^{T} 
  \gamma^k R_j(s_k, a_k, s_{k+1}) \biggr] \! \leq \!
  0, \
  j\in[q]. 
  \notag
\end{align}
Problem~\eqref{eq:safe-RL-problem-parameterized-policies} seeks to
find the policy $\pi_{\theta}$ that maximizes the expected cumulative
reward given by $R_0$ (for convenience, we have changed the
sign of $R_0$ to
turn~\eqref{eq:safe-RL-problem-parameterized-policies} into a
minimization problem) over $T$ time steps and also maintains the
expected cumulative costs given by $R_j$ for all $j\in[q]$ 
over $T$ time steps below zero.  Throughout the paper, we refer to 
$V_0, \hdots, V_q$ as \textit{value functions}. 
The discount factor~$\gamma$ determines how much future rewards are valued
compared to immediate rewards.

\begin{remark}\longthmtitle{Ensuring safety of  state
    trajectories}\label{rem:interpretation-safety} {\rm Throughout the
    paper, the notion of \emph{safety} refers to the satisfaction of
    the constraints
    in~\eqref{eq:safe-RL-problem-parameterized-policies}, and
    therefore pertains the policy parameter~$\theta$.  Interestingly,
    with an appropriate selection of the cost function $R_j$, this
    safety guarantee implies the forward invariance of a desired set
    $\Cc_j\subset \Sc$ with a prescribed confidence. In fact, let
    \[
      R_j(s_t) = 1 - \mathbbm{1}_{\Cc_j}(s_t) + \frac{\gamma^T
        \delta_j}{\sum_{t=0}^{T-1} \gamma^t},
    \]
    where
    $0 < \delta_j < 1$, for all $j\in [q]$, are prescribed confidence
    levels. According to~\cite[Theorems~1 and~2]{SP-MCF-LFOC-AR:23},
    the satisfaction of the cumulative constraints
    in~\eqref{eq:safe-RL-problem-parameterized-policies} implies
    that
    \[
      \mathbb{P} \Big( \bigcap_{t=0}^{T-1} \{ s_t \in \Cc_j \} \Big)
      \geq 1 - \delta_j, \quad \forall j \in [q],
    \]
    i.e., the probability that the states remain within $\Cc_j$ in the
    next $T$ timesteps is at least $1 - \delta_j$.  \demo}
\end{remark}

%
%

The functions $\{ V_i \}_{i=0}^q$ are in general non-convex, and this
makes solving~\eqref{eq:safe-RL-problem-parameterized-policies}
NP-hard.  Therefore, we aim to find local minimizers (or, more
generally, KKT points)
of~\eqref{eq:safe-RL-problem-parameterized-policies}.
Additionally, because of their definition, the values of
$V_0, \hdots, V_q$ and their gradients at arbitrary $\theta\in\real^d$
are not readily available, and instead need to be estimated through
episodic data (i.e., trajectories generated by the policy
$\pi_{\theta}$) of the CMDP.
%
%

Formally, we seek to solve the following problem.

\begin{problem}\label{problem:problem}
  Develop an RL algorithm that,
  \begin{itemize}
  \item converges to a KKT point
    of~\eqref{eq:safe-RL-problem-parameterized-policies};
  \item is anytime, meaning that at every iteration, the constraints
    of~\eqref{eq:safe-RL-problem-parameterized-policies} are satisfied.
  \end{itemize}
\end{problem}

Due to the probabilistic nature of the CMDP dynamics,
Problem~\ref{problem:problem} can only be solved in a probabilistic
sense, i.e., given a finite number of available episodes, one can only
expect to obtain convergence and constraint satisfaction results that
hold in probability.  As the number of available episodes grows, one
can also expect that the convergence and constraint satisfaction
guarantees hold with arbitrarily high probability.
%
%

%
%
%

\section{The Robust Safe Gradient Flow}\label{sec:robust-sgf}

Here we describe the Robust Safe Gradient Flow (RSGF), a
continuous-time anytime algorithm for constrained optimization that is
a variation of the \textit{Safe Gradient
  Flow}~\cite{AA-JC:24-tac,AA-JC:24-unpublished}.
%
%
We later rely on the RSGF to design our solution to
Problem~\ref{problem:problem}. Even though our proposed RL algorithm
will eventually be defined in discrete time, the properties of the
continuous-time flow established here are key, as we will leverage
them using the theory of stochastic approximation,
cf.~\cite{HJK-DSC:78,VSB:08}.

Let $V_0, \hdots, V_{\tilde{q}} :\real^d \to\real$ be continuously
differentiable functions and consider the constrained optimization
problem
\begin{align}\label{eq:optimization-problem-V0-to-Vqtilde}
  &\min_{\theta\in\real^d} \ V_0(\theta)
  \\  
  & \; \text{s.t.} \quad V_j(\theta) \leq 0, \
  j\in[\tilde{q}]. 
  \notag
\end{align}

We let
$\Cc = \setdef{\theta\in\real^d}{V_j(\theta)\leq0, \ \forall
  j\in[\tilde{q}]}$ denote the feasible set. Given $\alpha>0$ and a
continuously differentiable function $\beta:\real^d\to\real_{>0}$,
let  $\Rc_{\alpha,\beta}:\real^d\to\real^d$ be defined~by
\begin{align}\label{eq:R-definition}
  \Rc_{\alpha,\beta}(\theta)
  &=
    \argmin{\xi\in\real^d}{\frac{1}{2}\norm{\xi+\nabla
    V_0(\theta)}^2 }
  \\
  &\text{s.t.} \ \alpha V_j(\theta) + \nabla V_j(\theta)^\top \xi +
    \frac{\beta(\theta)}{2}\norm{\xi}^2 \leq 0, \ j\in[\tilde{q}].
      \notag
\end{align}
We note that if $\beta \equiv 0$, this definition recovers the Safe
Gradient Flow~\cite{AA-JC:24-tac}.  We study the properties of the
flow
\begin{align}\label{eq:ode-robust-safe-gradient-flow}
  \dot{\theta} = \Rc_{\alpha,\beta}(\theta), 
\end{align}
which we refer to as the Robust Safe Gradient Flow (RSGF). In
particular, we seek to determine conditions under which the dynamics
is well-posed and characterize the transient and asymptotic behavior
of its trajectories.

\subsection{Well-Posedness and Regularity Properties}

We start by introducing some regularity and constraint qualification
assumptions regarding the optimization
problem~\eqref{eq:optimization-problem-V0-to-Vqtilde}.

\begin{assumption}\longthmtitle{Regularity}\label{as:regularity-functions-Vi}
  The functions $V_0, \hdots, V_{\tilde{q}} :\real^d \to\real$, and
  $\beta:\real\to\real$ are twice continuously differentiable.
\end{assumption}

\begin{assumption}\longthmtitle{Constraint qualifications in the
    feasible set}\label{as:constraint-qualifications-safe-set}
  For all
  $\theta\in\Cc$,~\eqref{eq:optimization-problem-V0-to-Vqtilde}
  satisfies MFCQ.  Additionally, for each $\theta\in\Cc$, the
  parametric problem~\eqref{eq:R-definition} satisfies CRC at
  $(\theta,\Rc_{\alpha,\beta}(\theta))$.
\end{assumption}

\begin{assumption}\longthmtitle{Constraint qualifications outside the
    feasible set}\label{as:constraint-qualifications-outside-safe-set}
  For all $\theta\in\real^d\backslash\Cc$, Slater's condition holds
  for~\eqref{eq:R-definition} and the parametric
  problem~\eqref{eq:R-definition} satisfies CRC at
  $(\theta,\Rc_{\alpha,\beta}(\theta))$.
\end{assumption}

Assumption~\ref{as:regularity-functions-Vi} is standard in the
literature~\cite{KZ-AK-HZ-TB:20} and is satisfied by considering
smooth policies $\pi_{\theta}$.  MFCQ and CRC in
Assumptions~\ref{as:constraint-qualifications-safe-set},~\ref{as:constraint-qualifications-outside-safe-set}
are standard constraint qualification conditions for constrained
optimization problems
%
%
such as~\eqref{eq:optimization-problem-V0-to-Vqtilde}
and~\eqref{eq:R-definition}, and ensure that $\Rc_{\alpha,\beta}$
enjoys good regularity properties, as we establish in the sequel.
Lemma~\ref{lem:Slaters} provides conditions under which Slater's
condition holds for~\eqref{eq:R-definition} for each
$\theta\in\real^d\backslash\Cc$, and Lemma~\ref{lem:crc} provides
conditions under which CRC holds for~\eqref{eq:R-definition} at
$(\theta,\Rc_{\alpha,\beta}(\theta))$ for some $\theta\in\Cc$.
%
%
%
%

The next result provides a closed-form expression for
$\Rc_{\alpha,\beta}$ in terms of the Lagrange multipliers
of~\eqref{eq:R-definition}.

\begin{lemma}\longthmtitle{Alternative
    expression for RSGF}\label{lem:alternative-form}
  Let $u_j:\real^d \to \real$ map $\theta\in\real^d$ to the Lagrange
  multiplier associated with the $j$-th constraint
  of~\eqref{eq:R-definition}.  If MFCQ holds
  for~\eqref{eq:optimization-problem-V0-to-Vqtilde} at $\theta\in\Cc$,
  \begin{align}\label{eq:R-expression}
    \Rc_{\alpha,\beta}(\theta) = -\frac{\nabla V_0(\theta) +
    \sum_{j=1}^{\tilde{q}} u_j(\theta)\nabla
    V_j(\theta)}{1+\beta(\theta)\sum_{j=1}^{\tilde{q}} u_j(\theta)}.
  \end{align}
\end{lemma}
\begin{proof}
  Note that since $\theta\in\Cc$,
  $[\tilde{q}]=I_0(\theta)\cup I_{-}(\theta)$.  
  Since MFCQ holds
  for~\eqref{eq:optimization-problem-V0-to-Vqtilde} at $\theta$, there
  exists $\xi\in\real^d$ such that $\nabla V_j(\theta)^\top \xi < 0$
  for all $j\in I_0(\theta)$.  Hence, by taking
  $\epsilon_j < \frac{2|\nabla V_j(\theta)^\top \xi|}{\beta(\theta)
    \norm{\xi}^2}$ and $\hat{\xi} = \epsilon \xi$ with
  $\epsilon\in(0,\min\limits_{j\in I_0(\theta)} \epsilon_j )$,
  \begin{align*}
    \alpha V_j(\theta) + \nabla V_j(\theta)^\top \hat{\xi} +
    \frac{\beta(\theta)}{2} \norm{\hat{\xi}}^2 < 0, \quad \forall j\in
    I_0(\theta). 
  \end{align*}
  On the other hand, for every $j\in I_{-}(\theta)$, let $\epsilon_j$
  be sufficiently small so that
  $\alpha V_j(\theta) + \epsilon_j \nabla V_j(\theta)^\top \hat{\xi}
    + \epsilon_j^2 \frac{\beta(\theta)}{2}\norm{\hat{\xi}}^2 < 0$.
  Now, taking $\epsilon\in(0,\min\limits_{j\in[\tilde{q}] } \epsilon_j)$ and
  $\tilde{\xi} = \epsilon \xi$, we conclude
  $ \alpha V_j(\theta) + \nabla V_j(\theta)^\top \tilde{\xi} +
  \frac{\beta(\theta)}{2}\norm{\tilde{\xi}}^2 < 0$, for all
  $ j\in[\tilde{q}]$, and hence Slater's condition holds
  for~\eqref{eq:R-definition}.
  Since~\eqref{eq:R-definition} is convex,
  this means that $\Rc_{\alpha,\beta}$ satisfies the KKT equations
  associated to~\eqref{eq:R-definition}.
  Hence,
  \begin{align*}
    \Rc_{\alpha,\beta}(\theta) \!+\! \nabla V_0(\theta) \!+\! \sum_{j=1}^{\tilde{q}}
    u_j(\theta) \Big( \nabla V_i(\theta) +
    \beta(\theta)\Rc_{\alpha,\beta}(\theta) \Big) \! =\! 0, 
  \end{align*}
  from where the expression~\eqref{eq:R-expression} follows.
\end{proof}

The next result provides conditions under
which~\eqref{eq:R-definition} is feasible and locally Lipschitz.

\begin{lemma}\longthmtitle{Feasibility and
    Lipschitzness}\label{lem:feas-Lipschitzness-G}
  Suppose Assumption~\ref{as:regularity-functions-Vi}
  holds. Then,
  \begin{enumerate}
  \item\label{it:feas-lipsch-C} under
    Assumption~\ref{as:constraint-qualifications-safe-set},
    $\Rc_{\alpha,\beta}$ is well-defined and locally Lipschitz on an
    open neighborhood containing $\Cc$;
  \item\label{it:feas-lipsch-global} under
    Assumptions~\ref{as:constraint-qualifications-safe-set}
    and~\ref{as:constraint-qualifications-outside-safe-set},
    $\Rc_{\alpha,\beta}$ is well-defined and locally Lipschitz on
    $\real^d$.
  \end{enumerate}
\end{lemma}
%
%
\begin{proof}
  \ref{it:feas-lipsch-C}: By the argument employed in the proof of
  Lemma~\ref{lem:alternative-form}, Slater's condition holds
  for~\eqref{eq:R-definition} at any $\theta\in\Cc$.  This means that
  there exists $\xi\in\real^d$ such that
  $\alpha V_j(\theta) + \nabla V_j(\theta)^\top \xi +
  \frac{\beta(\theta)}{2}\norm{\xi}^2 < 0$ for all $j\in[q]$.  Since
  $V_j, \nabla V_j$ and $\beta$ are continuous, there exists a
  neighborhood $U_{\theta}$ of $\theta$ such that
  $\alpha V_j(\bar{\theta}) + \nabla V_j(\bar{\theta})^\top \xi +
  \frac{\beta(\bar{\theta})}{2}\norm{\xi}^2 < 0$ for all
  $\bar{\theta}\in U_{\theta}$.  In particular, the constraints in the
  definition of $\Rc_{\alpha,\beta}$ are feasible at all points in
  $U_{\theta}$ and hence $\Rc_{\alpha,\beta}$ is well-defined at all
  points in $U_{\theta}$.  Hence $\Rc_{\alpha,\beta}$ is well-defined
  in the open set $\cup_{\theta\in\Cc} U_{\theta}$ containing $\Cc$.
  Since SC implies MFCQ for convex problems~\cite[Proposition
  5.39]{NA-AE-MP:20}, the functions $V_0,\hdots,V_q$, and $\beta$ are
  twice continuously differentiable, and for each
  $\theta\in\Cc$,~\eqref{eq:R-definition} satisfies CRC at
  $(\theta,\Rc_{\alpha,\beta}(\theta))$, $\Rc_{\alpha,\beta}$ is
  locally Lipschitz on an open neighborhood of $\Cc$,
  invoking~\cite[Theorem 3.6]{JL:95}.

  \ref{it:feas-lipsch-global}: by assumption, for any
  $\theta\in\real^d$, Slater's condition holds
  for~\eqref{eq:R-definition} and CRC holds at
  $(\theta,\Rc_{\alpha,\beta}(\theta))$ for~\eqref{eq:R-definition}.
  Hence, by~\cite[Theorem~3.6]{JL:95}, $\Rc_{\alpha,\beta}$ is locally
  Lipschitz at~$\theta$.
\end{proof}

The local Lipschitzness of $\Rc_{\alpha,\beta}$ on a neighborhood of
$\Cc$, cf. Lemma~\ref{lem:feas-Lipschitzness-G}\ref{it:feas-lipsch-C}
(resp., in all of $\real^d$, cf.
Lemma~\ref{lem:feas-Lipschitzness-G}\ref{it:feas-lipsch-global})
ensures~\eqref{eq:ode-robust-safe-gradient-flow} is well-defined and
has a unique solution for any initial condition in a neighborhood of
$\Cc$ (resp., in all of $\real^d$). We refer to~\cite{PM-AA-JC:25-ejc}
for other conditions that guarantee local Lipschitzness of parametric
optimization problems such as~\eqref{eq:R-definition}.

\subsection{Equilibria, Forward Invariance, and Stability}

Next we establish the equivalence between the equilibrium points
of~\eqref{eq:ode-robust-safe-gradient-flow} and the KKT points
of~\eqref{eq:optimization-problem-V0-to-Vqtilde}.

\begin{proposition}\longthmtitle{Equivalence between equilibria and
    KKT points}\label{prop:equivalence-equilibria-kkt-points}
  Let~\eqref{eq:R-definition} be feasible at $\theta^*\in\real^d$.  If
  $\Rc_{\alpha,\beta}(\theta^*) = 0$, then $\theta^*\in\Cc$.  If MFCQ
  holds for~\eqref{eq:optimization-problem-V0-to-Vqtilde} at
  $\theta^*\in\real^d$, then $\Rc_{\alpha,\beta}(\theta^*) = 0$ if and
  only if $\theta^*$ is a KKT point
  of~\eqref{eq:optimization-problem-V0-to-Vqtilde}.
\end{proposition}
\begin{proof}
  If~\eqref{eq:R-definition} is feasible at $\theta^*\in\real^d$ and
  $\Rc_{\alpha,\beta}(\theta^*) = 0$, then
  $ \alpha V_j(\theta^*) + \nabla V_j(\theta^*)^\top
  \Rc_{\alpha,\beta}(\theta^*) + \frac{\beta(\theta^*)}{2}\norm{
    \Rc_{\alpha,\beta}(\theta^*) }^2 = \alpha V_j(\theta^*) \leq 0 $,
  for all $j\in[\tilde{q}]$, and therefore $\theta^* \in \Cc$.
  Next, suppose MFCQ holds
  for~\eqref{eq:optimization-problem-V0-to-Vqtilde} at
  $\theta^*\in\real^d$ and $\Rc_{\alpha,\beta}(\theta^*) = 0$.  As
  shown in the proof of Lemma~\ref{lem:alternative-form}(i), Slater's
  condition holds for~\eqref{eq:R-definition}.  Hence, since
  $\Rc_{\alpha,\beta}$ is the local minimizer, it satisfies the KKT
  equations for~\eqref{eq:R-definition}. Enforcing that the solution
  is $\xi=0$, these read exactly as the KKT equations
  for~\eqref{eq:optimization-problem-V0-to-Vqtilde}.
  %
  %
  Since MFCQ holds for~\eqref{eq:optimization-problem-V0-to-Vqtilde},
  it follows that $\theta^*$ is a KKT point
  of~\eqref{eq:optimization-problem-V0-to-Vqtilde}. Conversely, if
  $\theta^*$ is a KKT point
  of~\eqref{eq:optimization-problem-V0-to-Vqtilde}, then there exist a
  Lagrange multiplier vector $ u \in \real^{\tilde{q}}$ satisfying the
  KKT equations. Since the solution of~\eqref{eq:R-definition} is
  unique because the problem is strongly convex,
  %
  %
  we conclude $\Rc_{\alpha,\beta}(\theta^*) = 0$.
\end{proof}

The next result shows that $\Cc$ is forward invariant
under~\eqref{eq:ode-robust-safe-gradient-flow}.

\begin{proposition}\longthmtitle{Safety of RSGF}\label{prop:safety-rsgf}
  Suppose  Assumptions~\ref{as:regularity-functions-Vi}
  and~\ref{as:constraint-qualifications-safe-set} hold. Then, $\Cc$ is
  forward invariant under~\eqref{eq:ode-robust-safe-gradient-flow}.
\end{proposition}
\begin{proof}
  By Lemma~\ref{lem:feas-Lipschitzness-G}(i), every solution
  of~\eqref{eq:ode-robust-safe-gradient-flow} with initial condition
  in $\Cc$ is unique and well-defined as long as it stays in a
  neighborhood of~$\Cc$.
  Due to the constraints in~\eqref{eq:R-definition},
  \begin{align}\label{eq:exp-decrease}
    \nabla V_j(\theta)^\top \Rc_{\alpha,\beta}(\theta) \leq -\alpha
    V_j(\theta)-\frac{\beta(\theta)}{2}\norm{\Rc_{\alpha,\beta}(\theta)}^2,  
  \end{align}
  and hence $j\in[\tilde{q}]$,
  $\nabla V_j(\theta)^\top \Rc_{\alpha,\beta}(\theta) \leq 0$ whenever
  $V_j(\theta) = 0$.
        %
  %
  The result then follows from Nagumo's Theorem~\cite{MN:42}.
\end{proof}

The final result of this section characterizes the convergence
properties of~\eqref{eq:ode-robust-safe-gradient-flow}.

\begin{proposition}\longthmtitle{Convergence of RSGF}\label{prop:convergence}
  Suppose Assumptions~\ref{as:regularity-functions-Vi}
  and~\ref{as:constraint-qualifications-safe-set} hold. Then,
  \begin{enumerate}
  \item\label{it:convergence-C} every bounded trajectory
    of~\eqref{eq:ode-robust-safe-gradient-flow} starting in $\Cc$
    converges to the set of KKT points
    of~\eqref{eq:optimization-problem-V0-to-Vqtilde}.
  \item\label{it:convergence-global} if
    Assumption~\ref{as:constraint-qualifications-outside-safe-set}
    holds, then every bounded trajectory
    of~\eqref{eq:ode-robust-safe-gradient-flow} converges to the set
    of KKT points of~\eqref{eq:optimization-problem-V0-to-Vqtilde}.
  \end{enumerate}
  In either case, if every KKT point is isolated, convergence is to a
  point.
\end{proposition}
%
%
\begin{proof}
  \ref{it:convergence-C}: From the proof of
  Lemma~\ref{lem:alternative-form}(i), we have that
  $\Rc_{\alpha,\beta}(\theta)$ satisfies the KKT equations
  for~\eqref{eq:R-definition},
  \begin{align*}
    &u_j(\theta) \big( \alpha V_j(\theta) \!+\! \nabla
      V_j(\theta)^\top \! \Rc_{\alpha,\beta}(\theta) \!+\!
      \tfrac{\beta(\theta)}{2} \norm{\Rc_{\alpha,\beta}(\theta)}^2
      \big) = 0 , 
    \\ 
    &u_j(\theta) \!\geq\! 0, \ \alpha V_j(\theta) \!+\! \nabla
      V_j(\theta)^\top \Rc_{\alpha,\beta}(\theta) \!+\!
      \tfrac{\beta(\theta)}{2} \norm{\Rc_{\alpha,\beta}(\theta)}^2
      \leq 0. \notag
  \end{align*}
  Hence,
  \begin{align}\label{eq:dV0-dt}
    &\frac{d}{dt}V_0(\theta) = \nabla V_0(\theta)^\top
      \Rc_{\alpha,\beta}(\theta)=
    \\  
    \notag
    &\! -\Rc_{\alpha,\beta}(\theta)^\top \Big( \! \big(1 \! + \!
      \beta(\theta) \sum_{j=1}^q u_j(\theta)\!\big) 
      \Rc_{\alpha,\beta}(\theta) \! + \! \sum_{j=1}^q u_j(\theta)
      \nabla V_j(\theta) \! \Big)
    \\
    &= -\big(1+\tfrac{\beta(\theta)}{2}\sum_{j=1}^q
      u_j(\theta)\big)\norm{\Rc_{\alpha,\beta}(\theta)}^2 + \sum_{j=1}^q
      \alpha u_j(\theta) V_j(\theta),      \notag 
  \end{align}
  where in the second equality we have used~\eqref{eq:R-expression}
  and in the third we have used the KKT equations above. Now, since
  $u_j(\theta) \geq 0$ for all $j\in[q]$, and $V_j(\theta) \leq 0$ for
  all $j\in[q]$ if $\theta\in\Cc$, we deduce that
  $\frac{d}{dt}V_0(\theta)\leq 0$ for all $\theta\in\Cc$, with
  equality if and only if $\theta$ is a KKT point
  of~\eqref{eq:optimization-problem-V0-to-Vqtilde} by
  Proposition~\ref{prop:equivalence-equilibria-kkt-points}.  The fact
  that all bounded trajectories converge to the set of KKT points
  follows then from~\cite[Proposition 5.3]{SPB-DSB:03} using $V_0$ as a
  LaSalle function.
  Convergence to a point when the KKT points are isolated follows
  from~\cite[Corollary 5.2]{SPB-DSB:03}.

  \ref{it:convergence-global}: Let $\epsilon>0$ and define
  $V_{\epsilon_*}:\real^d\to\real$,
  \begin{align*}
    V_{\epsilon_*}(\theta) = V_0(\theta) +
    \frac{1}{\epsilon_*}\sum_{j=1}^q [V_j(\theta)]_{+}.
  \end{align*}
  %
  %
  From~\cite[Proposition 3]{GdP-LG:89}, $V_{\epsilon_*}$ is
  directionally differentiable and its directional
  derivative in the direction $\xi\in\real^n$ is
  \begin{align}\label{eq:dVepsilon-dt}
    \notag
    V_{\epsilon_*}^{\prime}(\theta;\xi)
    &= \nabla V_0(\theta)^\top \xi +
      \frac{1}{\epsilon_*}\sum_{j\in
      I_{+}(\theta)}\nabla
      V_j(\theta)^\top \xi
    \\ 
    & \quad + \frac{1}{\epsilon_*} \sum_{j\in I_{0}(\theta)}[\nabla
      V_j(\theta)^\top \xi]_{+}, 
  \end{align}
  where $I_0(\theta)$ and $I_{+}(\theta)$ correspond to the
  optimization problem~\eqref{eq:optimization-problem-V0-to-Vqtilde}.
  From the KKT equations above, we have that
  $\nabla V_j(\theta)^\top \Rc_{\alpha,\beta}(\theta) \leq -\alpha
  V_j(\theta)$ for all $j\in[q]$.  Using~\eqref{eq:dV0-dt}
  in~\eqref{eq:dVepsilon-dt} for $\xi = \Rc_{\alpha,\beta}(\theta)$,
  we have
  \begin{align*}
    V_{\epsilon_*}^{\prime}(\theta;\Rc_{\alpha,\beta}(\theta))
    &\leq 
      -\big( 1+\tfrac{\beta(\theta)}{2}\sum_{j=1}^q
      u_j(\theta) \big)\norm{\Rc_{\alpha,\beta}(\theta)}^2
    \\
    &\quad + \sum_{j=1}^q \alpha u_j(\theta)
      V_j(\theta)-\frac{1}{\epsilon_*}\sum_{j\in I_{+}(\theta)}\alpha
      V_j(\theta). 
  \end{align*}
  Now, by an argument analogous to~\cite[Lemma D.1]{AA-JC:24-tac},
  for any compact set $\Omega$, there exists $B_{\Omega} > 0$ such that 
  %
  %
  $u_j(\theta)\leq B_{\Omega}$ for all $j\in[\tilde{q}]$ and
  $\theta\in\Omega$. Then, for $\epsilon_* \in (0,\frac{1}{B_{\Omega}})$, and
  since $u_j(\theta) V_j(\theta) \leq 0$ for all
  $j\in I_{0}(\theta)\cup I_{-}(\theta)$, we have
  \begin{align*}
    V_{\epsilon_*}^{\prime}(\theta;\Rc_{\alpha,\beta}(\theta)) \! \leq \!
    -\big(1+\tfrac{\beta(\theta)}{2} \! \sum_{j=1}^q u_j(\theta)\! \big)
    \norm{\Rc_{\alpha,\beta}(\theta)}^2 \! \leq \! 0, 
  \end{align*}
  for all $\theta\in\Omega$,
  where the last inequality is an equality if and only if $\theta$ is
  a KKT point of~\eqref{eq:optimization-problem-V0-to-Vqtilde}
  (cf. Proposition~\ref{prop:equivalence-equilibria-kkt-points}).  Now
  the fact that all bounded trajectories in $\Omega$ converge to the
  set of KKT points of~\eqref{eq:optimization-problem-V0-to-Vqtilde} follows
  from~\cite[Proposition 5.3]{SPB-DSB:03} using $V_{\epsilon_*}$ as a
  LaSalle function. 
  Convergence to a point when the KKT points are isolated follows
  from~\cite[Corollary 5.2]{SPB-DSB:03}.
\end{proof}

\begin{remark}\longthmtitle{Boundedness of
    trajectories}\label{rem:boundedness} {\rm Regarding
    Proposition~\ref{prop:convergence}(i), note that all trajectories
    of~\eqref{eq:ode-robust-safe-gradient-flow} starting in $\Cc$
    remain in it by Proposition~\ref{prop:safety-rsgf}, and hence are
    bounded if this set is compact.  Regarding
    Proposition~\ref{prop:convergence}(ii), if there is
    $i_*\in[\tilde{q}]$ and $c$ such that
    $\Gamma = \setdef{\theta\in\real^d}{V_{i_*}(\theta) \leq c }$ is
    compact, note that~\eqref{eq:exp-decrease} implies that this set
    is forward invariant
    under~\eqref{eq:ode-robust-safe-gradient-flow}. Therefore, all
    trajectories of~\eqref{eq:ode-robust-safe-gradient-flow} starting
    in $\Gamma$ are bounded. In particular, this holds if $V_{i_*}$ is
    radially unbounded, since all its sublevel sets are compact.  }
  \demo
\end{remark}

\begin{remark}\longthmtitle{Robustness to
    error}\label{rem:robustness} {\rm The introduction of the
    strictly positive term $\beta$ in the
    definition~\eqref{eq:R-definition} strengthens the robustness
    against errors and disturbances of the robust safe gradient flow
    (as compared, for instance, with the safe gradient
    flow~\cite{AA-JC:24-tac}, which corresponds to $\beta \equiv 0$).
    An indication of this fact can be observed, for instance, in the
    contributions of the $\beta$ term to the decrease of the LaSalle
    functions in the proof of Proposition~\ref{prop:convergence}.
    We quantify more precisely this robustness to model errors in
    Section~\ref{sec:analysis-rl-rsgf} and exploit it to handle
    imperfect knowledge of the functions $\{ V_j \}_{j=0}^{\tilde{q}}$
    and their gradients $\{ \nabla V_j \}_{j=0}^{\tilde{q}}$ in the
    algorithm implementation.
    \demo }
\end{remark}

\begin{remark}\longthmtitle{Discretization}\label{rem:discretization}
  {\rm We note that the forward-Euler discretization
    of~\eqref{eq:ode-robust-safe-gradient-flow} is equivalent to the
    discrete-time dynamics introduced in~\cite{PM-AM-JC:25-l4dc}. This
    follows by performing a change of variables
    ($\xi = \frac{y-\theta}{h}$ in the optimization problem (2)
    in~\cite{PM-AM-JC:25-l4dc}, with the variables $y$ and $h$ as
    defined therein).  This discrete-time dynamics is a special case
    of the Moving Balls Algorithm (MBA)~\cite{AA-RS-MT:10}.
    Both~\cite{PM-AM-JC:25-l4dc} and~\cite{AA-RS-MT:10} study the
    safety and convergence properties of the discrete-time dynamics
    directly, instead of their continuous-time
    counterpart~\eqref{eq:ode-robust-safe-gradient-flow}, as we have
    done here. We leverage the latter in what follows using the theory
    of stochastic approximation~\cite{HJK-DSC:78,VSB:08}.  \demo}
\end{remark}

\section{Robust Safe Gradient Flow-Based Reinforcement
  Learning}\label{sec:proposed-solution}

In this section, we introduce our algorithmic solution to
Problem~\ref{problem:problem}.  Consider the optimization
problem~\eqref{eq:safe-RL-problem-parameterized-policies} defining the
optimal policy for the CMDP~$\mathcal{M}$. Instead of dealing directly
with~\eqref{eq:safe-RL-problem-parameterized-policies}, we
consider~\eqref{eq:optimization-problem-V0-to-Vqtilde} with
$\tilde{q} = q+1$, and include the additional function
$V_{q+1}(\theta) = \norm{\theta}^2 - C$, where $C>0$ is a design
parameter. As we justify later, this has the effect of keeping the
iterates of the algorithm bounded.

Given $\alpha>0$ and $\beta:\real^d\to\real_{>0}$, let
$\Rc_{\alpha,\beta}:\real^d\to\real^d$ be defined
by~\eqref{eq:R-definition}.  To solve Problem~\ref{problem:problem},
consider the forward-Euler discretization of the
RSGF~\eqref{eq:ode-robust-safe-gradient-flow},
\begin{align}\label{eq:discrete-time-rsgf}
  \theta_{i+1} = \theta_i + h_i \Rc_{\alpha,\beta}(\theta_i) ,
\end{align}
where $\{ h_i \}_{i\in\mathbb{Z}_{>0}}$ is a sequence of
stepsizes. Note that, since closed-from expressions for the value
functions $V_0, \hdots, V_q$ are not readily available, one cannot
directly implement this iteration. Instead, our strategy consists of
relying on the robustness properties of~\eqref{eq:discrete-time-rsgf},
when viewed as a discrete-time dynamical system, and employing estimates
of $V_1,\hdots,V_q$, and $\nabla V_0,\hdots,\nabla V_q$ constructed
with episodic data, as detailed next
(note that $V_{q+1}$ and $\nabla V_{q+1}$ are known).

\subsubsection*{Episodic data available}
Let $\Lambda$ be a given set of policies for $\mathcal{M}$ and
$ \Ic_0$ a batch of episodes obtained offline with policies
from~$\Lambda$. Formally,
\begin{align*}
  \Ic_0 = \{ [ s_0^n, a_0^n, s_1^n, a_1^n,
  \hdots, s_{T}^n, a_T^n, s_{T+1}^n ] \}_{
  n\in[N_{\zeta}], \zeta\in\Lambda }, 
\end{align*}
where $N_{\zeta}$ is the number of episodes obtained with policy
$\zeta$.  Given $i\in\mathbb{Z}_{>0}$, let $\Ic_i$ be the collection
of episodes at iteration $i$ obtained using policy $\pi_{\theta_i}$
(with $N_i = |\Ic_i| $ its number).

At iteration~$i$, we construct the estimates of the value functions
and their gradients using episodes from $\cup_{j=0}^{i} \Ic_j$ as
follows. Although one could potentially use all such episodes, for
flexibility we assume that we only use a subset
$\Jc_i \subset \cup_{j=0}^{i} \Ic_j$.  We enumerate the episodes in
$\Jc_i$ as
\begin{align*}
  \Jc_i = \{ [s_0^n, a_0^n, s_1^n, a_1^n, \hdots, s_T^n, a_T^n,
  s_{T+1}^n] \}_{n=1}^{|\Jc_i|}. 
\end{align*}
For each $n\in[|\Jc_i|]$, we denote by $\zeta_n$ the policy utilized
to obtain the corresponding episode.

\begin{assumption}\label{as:importance-sampling-ratio-well-defined}
  There exists $\nu>0$ such that, for any $a\in\Ac$, $s\in\Sc$,
  $\theta \in\real^d$ and $\zeta\in\Lambda$, we have
  $\pi_{\theta}(a|s)>\nu$, $\zeta(a|s)>\nu$.
\end{assumption}

Assumption~\ref{as:importance-sampling-ratio-well-defined} is standard
in the context of importance-sampling methods in
RL~\cite{TX-ZY-ZW-YL:21,JH-NJ:22}. For any given state, it requires
that any action has a positive probability lower bounded by $\nu$ for
any policy in the parametric family $\{ \pi_{\theta} \}$ as well as
in~$\Lambda$.
%
%

\subsubsection*{Estimates of value functions and their gradients}
For each $j\in[q]\cup\{ 0 \}$, we consider the following estimate of
the value function at iteration~$i$,
\begin{multline}\label{eq:value-function-estimate}
  \widehat{V_j}(\theta_i)
  =
  \\
  \frac{\sigma_j}{|\Jc_i|}
  \bigg(
  \sum_{n=1}^{|\Jc_i|} \prod\limits_{t=0}^T \frac{
    \pi_{\theta_i}(a_t^n |s_t^n) }{ \zeta_n(a_t^n |s_t^n) }
  \sum_{t=0}^T \gamma^t R_j( s_t^n, a_t^n, s_{t+1}^n ) 
  \bigg),
\end{multline}
where $\sigma_0 = -1$, and $\sigma_j = 1$ for $j\in[q]$.
Under
Assumption~\ref{as:importance-sampling-ratio-well-defined},
$\widehat{V_j}(\theta_i)$ is well defined, because the denominator in
the ratio
$\frac{ \pi_{\theta_i}(a_t^n |s_t^n) }{\zeta_n(a_t^n |s_t^n) }$ is strictly positive.
%
%

For any $a\in\Ac$ and $s\in\Sc$, define
$\chi_{a,s}:\real^d\to\real$ as
$\chi_{a,s}(\theta) = \log\pi_{\theta}(a|s)$.  
Note that $\chi_{a,s}$ is well-defined for all $\theta\in\real^d$ under 
Assumption~\ref{as:importance-sampling-ratio-well-defined}.
Let $b:\Sc\to\real$ be
a baseline function whose absolute value is bounded by $\hat{B}>0$.
For each $j\in[q]\cup\{ 0 \}$, we consider the following estimates of
the gradients of the value functions at iteration~$i$,
%
%
\begin{subequations}\label{eq:gradient-value-function-estimate}
  \begin{multline}
  \widehat{\nabla V_j}(\theta_i) =
  \\
  \frac{\sigma_j}{ |\Jc_i| } \bigg( \sum_{n=1}^{|\Jc_i|}
  \prod\limits_{t=0}^T \frac{ \pi_{\theta_i}(a_t^n |s_t^n) }{
    {\zeta_n}(a_t^n |s_t^n) } \sum_{t=0}^T \gamma^t \nabla
  \chi_{a_t^n,s_t^n}(\theta_i) D_{j,t}^n \bigg).
\end{multline}
%
%
where
\begin{align}
  D_{j,t}^n = \sum_{t^{\prime}=t}^{T} \! \gamma^{t^{\prime}-t}
  R_j(s_{t^\prime}^n,a_{t^{\prime}}^n,s_{t^{\prime}+1}^n) \! - \! b(s_t^n).
\end{align}
\end{subequations}
Under Assumption~\ref{as:importance-sampling-ratio-well-defined},
$ \widehat{\nabla V_j}(\theta_i) $ is well defined.
Given these estimates, we define an approximated version
of~\eqref{eq:R-definition} as follows:
\begin{subequations}
  \begin{align}
    &\hat{\Rc}_{\alpha,\beta}(\theta) =
      \argmin{\xi \in\real^d}{ \frac{1}{2}\norm{\xi + \widehat{\nabla
      V_0}(\theta)}^2 } 
    \\ 
    &\quad \text{s.t.} \ \alpha \widehat{V}_j(\theta) \! + \!
      \widehat{\nabla V}_j(\theta)^\top \xi \! 
      + \! \frac{\beta(\theta)}{2}\norm{\xi}^2 \leq 0, \ j \in[q],
    \\
    &\quad \quad \ \ \alpha V_{q+1}(\theta) \! + \! \nabla
      V_{q+1}(\theta)^\top \xi \! + \!
      \frac{\beta(\theta)}{2}\norm{\xi}^2 \leq 0. \label{eq:Vq+1-Rhat}
  \end{align}
  \label{eq:Rhat-definition}
\end{subequations}
Note that this can be computed with the episodic data available to the agent.

Algorithm~\ref{alg:rl-rsgf} presents the pseudocode for our proposal
to solve Problem~\ref{problem:problem}. We refer to it as Robust Safe
Gradient Flow-based Reinforcement Learning (RSGF-RL).
%
%

\begin{algorithm}[htb!]
  \caption{\texttt{RSGF-RL}}\label{alg:rl-rsgf}
  \begin{algorithmic}[1]
    \State \textbf{Parameters}: $\alpha$, $\beta$, $k$, $m$, $\{ h_i
    \}_{i=1}^k$ $T$, $\gamma$, $\Ic_0$, $\{ N_i \}_{i=1}^k$
    \State \textbf{Initial Policy Parameter}: $\theta_1$
    \For{$i\in [k]$}
      \State Generate $N_i$ episodes of length $T+1$ using $\pi_{\theta_i}$
      \State Select the set $\Jc_i$ of episodes at iteration~$i$
      \State Compute estimates $\{\widehat{V_j}(\theta_i)\}_{j=0}^q$ 
      using~\eqref{eq:value-function-estimate}
      \State Compute estimates 
      $\{\widehat{\nabla V_j}(\theta_i)\}_{j=0}^q$
      using~\eqref{eq:gradient-value-function-estimate} 
      \State Update policy according to
      \begin{align}\label{eq:discrete-time-rsgf-estimates}
        \theta_{i+1}=\theta_i + h_i \hat{\Rc}_{\alpha,\beta}(\theta_i)
      \end{align}
    \EndFor
    \State \Return $\theta_{k+1}$
\end{algorithmic}
\end{algorithm}

In Algorithm~\ref{alg:rl-rsgf}, we do not detail a specific scheme to
select the sets of episodes $\Jc_i$ from the available ones in
$\cup_{j=0}^{i} \Ic_j$. Instead, in what follows, we study the
properties of RSGF-RL for arbitrary sets~$\Jc_i$ and provide
conditions on these sets that guarantee a desired level of algorithmic
performance.

%
%

%
%

\section{Anytime Safety and Convergence Guarantees of
  RSGF-RL}\label{sec:analysis-rl-rsgf}

In this section we present our technical analysis of RSGF-RL. We start
by establishing different statistical properties of the value function
and gradient estimates, and then characterize the safety and
convergence properties of RSGF-RL.

\subsection{Statistical Properties of Estimates
}

Here, we establish the statistical properties of the
estimates~\eqref{eq:value-function-estimate}
and~\eqref{eq:gradient-value-function-estimate} of the value functions
and their gradients, resp.  In our analysis, we make the following
assumptions.

\begin{assumption}\longthmtitle{Boundedness of reward
    functions}\label{as:boundedness-reward}
  For each $j\in[q]\cup\{ 0 \}$, there exist $B_j>0$ such that
  $|R_j(s,a,s^{\prime})|<B_j$, for all $s\in\mathcal{S}$,
  $a\in\mathcal{A}$, and $s^{\prime}\in\mathcal{S}$.
\end{assumption}

\begin{assumption}\longthmtitle{Differentiability and Lipschitzness of
    policy}\label{as:differentiability-lipschitzness-policy}
  %
  %
  %
  %
  The function $\chi_{a,s}$ is continuously differentiable and there
  exist $L>0$ and $\tilde{B}>0$ such that
  \begin{align*}
    &\norm{\nabla \chi_{a,s}(\theta)-\nabla \chi_{a,s}(\bar \theta)}
      \leq L\norm{\theta-\bar\theta},
    \\
    &\qquad \qquad \qquad \qquad \qquad \qquad  \forall 
      \theta,\bar \theta\in\real^d, a\in\Ac,s\in\Sc, 
    \\
    & \norm {\nabla  \chi_{a,s}(\theta)^{(l)} }
      \leq \tilde{B}, \ \forall\theta\in\real^d, l\in[d],
      a\in\Ac,s\in\Sc. 
  \end{align*}
  %
  %
\end{assumption}

Assumptions~\ref{as:boundedness-reward}
and~\ref{as:differentiability-lipschitzness-policy} are standard in
the literature, cf.~\cite{KZ-AK-HZ-TB:20,QB-WUM-VA:24}. By the Policy
Gradient Theorem~\cite[Section 13.2]{RSS-AGB:18}, under
Assumption~\ref{as:differentiability-lipschitzness-policy},
%
%
the functions $\{ V_j \}_{j=0}^q$
in~\eqref{eq:safe-RL-problem-parameterized-policies} are
differentiable.  Moreover,
Lemma~\ref{lem:lipschitzness-value-functions} ensures that, for all
$j\in\{ 0 \} \cup [q]$, $\nabla V_j$ is globally Lipschitz on
$\real^d$ (we denote by $L_j$ its Lipschitz constant).  Additionally,
we let $L_{q+1} = 2\sqrt{C}$ be the Lipschitz constant of $V_{q+1}$ on
$\Theta = \setdef{\theta\in\real^d}{ \norm{\theta}^2 \leq C }$.

In what follows, all expectations, variances, and probabilities are
taken with respect to $s_0\sim\eta$,
$a_t^{n} \sim \zeta_n(\cdot|s_t^{n})$, for $t\in[T]$, and
$n\in[|\Jc_i|]$. We next characterize the mean, variance, and
tail probabilities of the value function estimates.

\begin{proposition}\longthmtitle{Value function
    estimates}\label{prop:value-function-estimates}
  Suppose Assumptions~\ref{as:importance-sampling-ratio-well-defined}
  and~\ref{as:boundedness-reward} hold.
  Let $i\in\mathbb{Z}_{>0}$ and assume that $\Jc_i$ contains
  $\bar{N}_i$ episodes generated with $\pi_{\theta_i}$ (without loss
  of generality, we label them as the first $\bar{N}_i$ episodes in
  $\Jc_i$).  Let
  \begin{align*}
    \tilde{N}_i = |\Jc_i|-\bar{N}_i , \;
    \phi_j = \frac{B_j(1-\gamma^{T+1})}{1-\gamma} , \;
    \bar{\phi}_j = \frac{B_j(1-\gamma^{T+1})}{(1-\gamma)\nu^{T+1}} .
  \end{align*}
  Then, for $j\in\{0\}\cup[q]$,
  %
  %
  \begin{enumerate}
  \item\label{it:value-function-unbiased}
    $\mathbb{E}[ \widehat{V}_j(\theta_i) ] = V_j(\theta_i)$ (unbiased
    function estimates);
  \item\label{it:value-function-variance}
    $\operatorname{Var}[ \widehat{V}_j(\theta_i) ] = \frac{ \bar{N}_i
      \phi_j^2 + \tilde{N}_i \bar{\phi}_j^2 }{ |\Jc_i|^2 }$ and
    $|\widehat{V}_j(\theta_i)| \leq \frac{ \bar{N}_i \phi_j +
      \tilde{N}_i \bar{\phi}_j }{ |\Jc_i| }$;
  \item\label{it:value-function-tail-probability}
    $\mathbb{P}( |\widehat{V}_j(\theta_i)-V_j(\theta_i)| \leq \epsilon
    ) \geq 1 - 2 \exp\big( -\frac{ \epsilon^2 |\Jc_i|^2 }{
      2\bar{N}_i\phi_j^2 + 2\tilde{N}_i \bar{\phi}_j^2 } \big)$.
  \end{enumerate}
  Further assume that $\chi_{a,s}$ is globally Lipschitz, uniformly in
  $a, s$, i.e., there exists $\tilde{L}>0$ such that
  \begin{align}\label{eq:extra-assumption-difference-log-policies}
    &|\chi_{a,s}(\theta) - \chi_{a,s}(\theta^{\prime})| \leq
      \tilde{L}\norm{\theta-\theta^{\prime}}, 
  \end{align}
  for all $ \theta, \theta^{\prime}\in\real^d$, $a\in\Ac$, and
  $s\in\Sc$, and that the policies in $\Lambda$ belong to
  $\{\pi_{\theta}\}$. Let $\{ \bar{\theta}_n \}_{n=1}^{|\Jc_i|}$
  denote the parameters that describe all the policies in $\Jc_i$ and
  define
  $\tilde{\phi}_{i,j,n} = \phi_j \exp(
  (T+1) \tilde{L}\norm{\theta_i-\bar{\theta}_n} )$.  Then,
  \begin{enumerate}
    \setcounter{enumi}{3}
  \item\label{it:value-function-variance-extra}
    $\VaR[ \widehat{V}_j(\theta_i) ] \leq \frac{ \sum_{n=1}^{|\Jc_i|}
      \tilde{\phi}_{i,j,n}^2 }{ |\Jc_i|^2 }$ and
    $|\widehat{V_j}(\theta_i)| \leq \frac{ \sum_{n=1}^{|\Jc_i|}
      \tilde{\phi}_{i,j,n} }{ |\Jc_i| }$;
    
  \item\label{it:value-function-tail-probability-extra}
    $\mathbb{P}( |\widehat{V}_j(\theta_i)-V_j(\theta_i)| \leq \epsilon
    ) \geq 1 - 2\exp\big( -\frac{ \epsilon^2 |\Jc_i|^2 }{
      2\sum_{n=1}^{|\Jc_i|} \tilde{\phi}_{i,j,n}^2 } \big)$.
  \end{enumerate}
\end{proposition}
\begin{proof}
  \ref{it:value-function-unbiased}: Let
  $d\Omega = \prod_{n=1}^{|\Jc_i|} ds_{T+1}^n \prod_{t=0}^{T} ds_t^n
  da_t^n$, where $ds_t^n$ and $d a_t^n$ are the differential elements
  associated with the variables $s_t^n$ and $a_t^n$, respectively.
  %
  %
  Define, for $j\in\{ 0 \}\cup[q]$,
  \begin{align*}
    E_{j,n}
     = \sum_{t^{\prime}=0}^T \gamma^{t} R_j(
      s_{t}^n, a_{t}^n, s_{t+1}^n ), \ 
    \Gamma = \Sc^{ |\Jc_i|(T+2)} \times \Ac^{ |\Jc_i|(T+1) },
  \end{align*}
  Using~\eqref{eq:value-function-estimate},
  we have
  \begin{align*}
    &\mathbb{E} [ \widehat{V}_j(\theta_i) ] = \frac{\sigma_j}{|\Jc_i|}
      \bigintsss_{ \Gamma }
      \bigg( \eta(s_0) \sum_{n=1}^{\bar{N}_i} \prod\limits_{t=0}^T
      \pi_{\theta_i}(a_t^n |s_t^n) E_{j,n}
    \\
    &\qquad \qquad +
      \eta(s_0)\sum_{n=\bar{N}_i+1}^{|\Jc_i|} \prod\limits_{t=0}^T
      \frac{ \pi_{\theta_i}(a_t^n |s_t^n) }{ \zeta_n(a_t^n |s_t^n)
      }E_{j,n} \zeta_n(a_t^n|s_t^n) 
      \bigg) d\Omega  
    \\
    &= \frac{1}{ \bar{N}_i + \tilde{N}_i }( \bar{N}_i V_j(\theta_i) +
      \tilde{N}_i V_j(\theta_i) ) = V_j(\theta_i). 
  \end{align*}

  \ref{it:value-function-variance}: By
  Assumption~\ref{as:importance-sampling-ratio-well-defined},
  $\zeta_n(a_t^n|s_t^n) > \nu$ for all $n\in[|\Jc_i|]$ and $t\in[T]$.
  This implies that, for each $n\in[\bar{N}_i : |\Jc_i|]$,
  \begin{multline}\label{eq:bound-value-function-off-policy-term}
    \bigg\rvert \prod\limits_{t=0}^T \frac{ \pi_{\theta_i}(a_t^n
      |s_t^n) }{ \zeta_n(a_t^n |s_t^n) } \bigg( \sum_{t=0}^T \gamma^t
    R_j( s_t^n, a_t^n, s_{t+1}^n ) \bigg) 
    \bigg\rvert 
    \\
    \leq B_j \frac{1-\gamma^{T+1}}{1-\gamma}\frac{1}{\nu^{T+1}} =
    \bar{\phi}_j 
  \end{multline}
  By Popovicius' inequality~\cite[Corollary 1]{RB-CD:00},
  %
  %
  we have
  \begin{align*}
    \VaR\biggl[ \prod\limits_{t=0}^T \frac{
    \pi_{\theta_i}(a_t^n |s_t^n) }{ \zeta_n(a_t^n |s_t^n) } \bigg(
    \sum_{t=0}^T \gamma^t R_j( s_t^n, a_t^n, s_{t+1}^n ) \bigg) 
    \biggr] \leq \bar{\phi}_j^2.
  \end{align*}
  Since the random variables
  \begin{align*}
    \bigg\{ \prod\limits_{t=0}^T \frac{ \pi_{\theta_i}(a_t^n |s_t^n)
    }{ \zeta_n(a_t^n |s_t^n) } \bigg( \sum_{t=0}^T \gamma^t R_j(
    s_t^n, a_t^n, s_{t+1}^n ) \bigg) \bigg\}_
    {
    n\in[\bar{N}_i : |\Jc_i|]
    }
  \end{align*}
  are independent, it follows that 
  \begin{align*}
    \VaR\biggl[ \sum_{n=\bar{N}_i+1}^{|\Jc_i|}
    \prod\limits_{t=0}^T \frac{ \pi_{\theta_i}(a_t^n |s_t^n) }{
    \zeta_n(a_t^n |s_t^n) } \bigg( \! \sum_{t=0}^T \gamma^t R_j(
    s_t^n, a_t^n, s_{t+1}^n ) \! \bigg) \biggr] \! \leq \! \tilde{N}_i
    \bar{\phi}_j^2. 
  \end{align*}
  On the other hand, note that
  \begin{align}\label{eq:bound-value-function-on-policy-term}
    \bigg\rvert \sum_{t=0}^T \gamma^t R_j(s_t^n,a_t^n,s_{t+1}^n)
    \bigg\rvert \leq B_j \frac{1-\gamma^{T+1}}{1-\gamma} = \phi_j. 
  \end{align}
  By Popovicius' inequality~\cite[Corollary 1]{RB-CD:00}, 
  \begin{align*}
    \VaR\biggl[ \sum_{\bar{n}=1}^{\bar{N}_i} \sum_{t=0}^T
    \gamma^t R_j(s_t^n,a_t^n,s_{t+1}^n) \biggr] \leq \bar{N}_i
    \phi_j^2, 
  \end{align*}
  from where the bound on the variance follows.  Note also
  that~\eqref{eq:bound-value-function-off-policy-term}
  and~\eqref{eq:bound-value-function-on-policy-term} imply that
  $\widehat{V_j}(\theta_i)$ is uniformly upper bounded by
  $\frac{\bar{N}_i \phi_j + \tilde{N}_i \bar{\phi}_j }{ |\Jc_i| }$.

  \ref{it:value-function-tail-probability}: This follows from
  Hoeffding's inequality~\cite{WH:63-hoeffding}
  %
  %
  using~\eqref{eq:bound-value-function-off-policy-term}
  and~\eqref{eq:bound-value-function-on-policy-term}.

  \ref{it:value-function-variance-extra}:
  Under~\eqref{eq:extra-assumption-difference-log-policies}, we have
  \begin{align*}
    \frac{ \pi_{\theta}(a |s) }{ \pi_{\theta^{\prime}}(a |s) } \leq
    \exp\big( \tilde{L}\norm{\theta-\theta^{\prime} } \big),
  \end{align*}
  for any $\theta,\theta^{\prime} \in\real^d$. Therefore,
  \begin{multline}\label{eq:bound-value-function-off-policy-extra-assumption}
    \bigg\rvert \prod\limits_{t=0}^T \frac{ \pi_{\theta_i}(a_t^n
      |s_t^n) }{ \pi_{\bar{\theta}_n}(a_t^n |s_t^n) } \bigg(
    \sum_{t=0}^T \gamma^t R_j( s_t^n, a_t^n, s_{t+1}^n ) \bigg) 
    \bigg\rvert
    \\
    \leq B_j \frac{1-\gamma^{T+1}}{1-\gamma}\exp( (T+1)
    \tilde{L}\norm{ \theta_i - \bar{\theta}_n } ) =
    \tilde{\phi}_{i,j,n}. 
  \end{multline}
  By Popovicius' inequality~\cite[Corollary 1]{RB-CD:00}, this implies
  \begin{align*}
    \VaR\biggl[ \prod\limits_{t=0}^T \frac{ \pi_{\theta_i}(a_t^n
    |s_t^n) }{ \pi_{\bar{\theta}_n}(a_t^n |s_t^n) } \bigg(
    \sum_{t=0}^T \gamma^t R_j( s_t^n, a_t^n, s_{t+1}^n ) \bigg)
    \biggr] \leq \tilde{\phi}_{i,j,n}^2. 
  \end{align*}
  The result now follows by noting that the random variables
  \begin{align*}
    \bigg\{ \prod\limits_{t=0}^T \frac{ \pi_{\theta_i}(a_t^n |s_t^n) }{
    \pi_{\bar{\theta}_n}(a_t^n |s_t^n) } \bigg( \sum_{t=0}^T \gamma^t
    R_j( s_t^n, a_t^n, s_{t+1}^n ) \bigg) \bigg\}_{n\in[|\Jc_i|]}    
  \end{align*}
  are independent.  Note also
  that~\eqref{eq:bound-value-function-off-policy-extra-assumption}
  implies that $\widehat{V_j}(\theta_i)$ is uniformly upper bounded by
  $\frac{\sum_{n=1}^{|\Jc_i|} \tilde{\phi}_{i,j,n}}{|\Jc_i|}$.

  \ref{it:value-function-tail-probability-extra}: This follows from
  Hoeffding's inequality~\cite{WH:63-hoeffding}
  using~\eqref{eq:bound-value-function-off-policy-extra-assumption}.
  %
  %
\end{proof}

In Proposition~\ref{prop:value-function-estimates}, the bounds
in~\ref{it:value-function-variance-extra}
and~\ref{it:value-function-tail-probability-extra} derived under the
additional
assumption~\eqref{eq:extra-assumption-difference-log-policies} are
tighter than the ones in~\ref{it:value-function-variance}
and~\ref{it:value-function-tail-probability}.  This is because
$\tilde{\phi}_{i,j,n}$, which appears
in~\ref{it:value-function-variance-extra}
and~\ref{it:value-function-tail-probability-extra}, depends on the
difference between the policy parameters $\bar{\theta}_n$ and
$\theta_i$, so it takes advantage of their proximity. Instead,
$\bar{\phi}_j$, which appears in~\ref{it:value-function-variance}
and~\ref{it:value-function-tail-probability}, is insensitive to this
proximity.  We next study the statistical properties of the gradients.

\begin{proposition}\longthmtitle{Gradient of value function
    estimates}\label{prop:gradient-value-function-estimates}
  Suppose
  Assumptions~\ref{as:importance-sampling-ratio-well-defined},~\ref{as:boundedness-reward},
  and~\ref{as:differentiability-lipschitzness-policy} hold.  Let
  $i\in\mathbb{Z}_{>0}$, $\bar{N}_i$ and $\tilde{N}_i$ as in
  Proposition~\ref{prop:value-function-estimates}, and
  \begin{align*}
    &\psi_j = \tilde{B}\sum_{t=0}^T \gamma^t \sum_{t^\prime = t}^T
      (\gamma^{t^{\prime}-t}B_j + \hat{B}),
    \\
    &\bar{\psi}_j = \frac{\tilde{B}}{\nu^{T+1}} \sum_{t=0}^T \gamma^t
      \sum_{t^\prime = t}^T (\gamma^{t^{\prime}-t}B_j + \hat{B}). 
  \end{align*}
  Then, for $j\in\{ 0 \} \cup [q]$
  \begin{enumerate}
  \item\label{it:gradient-unbiased}
    $\mathbb{E}[ \widehat{\nabla V}_j(\theta_i) ] = \nabla
    V_j(\theta_i)$ (unbiased gradient estimates);
  \item\label{it:gradient-variance}
    $\VaR[\widehat{\nabla V}_j(\theta_i)^{(l)} ] = \frac{ \bar{N}_i
      \psi_j^2 + \tilde{N}_i \bar{\psi}_j^2 }{ |\Jc_i|^2 }$ and
    $|\widehat{\nabla V}_j(\theta_i)^{(l)}| \leq \frac{ \bar{N}_i
      \psi_j + \tilde{N}_i \bar{\psi}_j }{ |\Jc_i| }$, for all
    $l\in[d]$;
  \item\label{it:gradient-tail-probability}
    $\mathbb{P}\big( \norm{\widehat{\nabla V}_j(\theta_i)-\nabla
      V_j(\theta_i)} \leq \epsilon \big) \geq 1 - 2 d \exp \big(
    -\frac{ \epsilon^2 |\Jc_i|^2 }{ 2d(\bar{N}_i\psi_j^2 + \tilde{N}_i
      \bar{\psi}_j^2) } \big)$.
  \end{enumerate}
  Further assume that $\chi_{a,s}$ is globally Lipschitz, uniformly in
  $a, s$, i.e.,~\eqref{eq:extra-assumption-difference-log-policies}
  holds, and that the policies in $\Lambda$ belong to
  $\{ \pi_{\theta} \}$.  Let $\{ \bar{\theta}_n \}_{n=1}^{|\Jc_i|}$
  denote the parameters that describe all the policies in $\Jc_i$ and
  define
  $\tilde{\psi}_{i,j,n} =
  \exp((T+1)\tilde{L}\norm{\theta_i-\bar{\theta}_n} )\psi_j$.
  %
  %
  Then,
  \begin{enumerate}
    \setcounter{enumi}{3}
  \item\label{it:gradient-variance-extra}
    $\VaR[\widehat{\nabla V}_j(\theta_i)^{(l)} ] \leq \frac{
      \sum_{n=1}^{|\Jc_i|} \tilde{\psi}_{i,j,n}^2 }{ |\Jc_i|^2 }$ and
    $|\widehat{\nabla V_j}(\theta_i)^{(l)}| \leq \frac{
      \sum_{n=1}^{|\Jc_i|} \tilde{\psi}_{i,j,n} }{ |\Jc_i| }$, for all
    $l\in[d]$;
    
  \item\label{it:gradient-tail-probability-extra}
    $\mathbb{P} \big( \norm{\widehat{\nabla V}_j(\theta_i)-\nabla
      V_j(\theta_i)} \leq \epsilon \big) \geq 1 - 2 d\exp\big( -\frac{
      \epsilon^2 |\Jc_i|^2 }{ 2\sum_{n=1}^{|\Jc_i|}
      \tilde{\psi}_{i,j,n}^2 } \big)$.
  \end{enumerate} 
\end{proposition}
\begin{proof}
  \ref{it:gradient-unbiased}: Let $j\in[q]\cup\{ 0 \}$.  With the
  notation of~\eqref{eq:gradient-value-function-estimate}, by the
  Policy Gradient Theorem with baseline (cf.~\cite[Section
  13.4]{RSS-AGB:18}), for each $n\in[\bar{N}_i]$, we have
  \begin{align}
    \notag
    &\mathbb{E} \biggl[ \sigma_j\sum_{t=0}^T
      \gamma^t \nabla \chi_{a_t^{n},s_t^{n}}(\theta_i)
      D_{j,t}^n
      \biggr] = \nabla V_j(\theta_i).  
  \end{align}
  On the other hand, for $n\in[ \bar{N}_i : |\Jc_i| ]$,
  %
  %
  \begin{align*}
    &\mathbb{E}\Biggl[ 
      \sigma_j \prod\limits_{t=0}^T \frac{ \pi_{\theta_i}(a_t^n
      |s_t^n) }{ \zeta_n(a_t^n |s_t^n) } \sum_{t=0}^T 
      \gamma^t \nabla \chi_{a_t^n,s_t^n}(\theta_i) D_{j,t}^n
      \Biggr]
      =
    \\ 
    &\bigintsss_{ \Gamma }
      \sigma_j \prod\limits_{t=0}^T \frac{ \pi_{\theta_i}(a_t^n
      |s_t^n) }{\zeta_n(a_t^n |s_t^n) } \sum_{t=0}^T 
      \gamma^t \nabla \chi_{a_t^n,s_t^n}(\theta_i) D_{j,t}^n
      \zeta_n(a_t^n|s_t^n) d\Omega
    \\
    &= \bigintsss_{ \Gamma }
      \sigma_j \eta(s_0)\prod\limits_{t=0}^T \pi_{\theta_i}(a_t^n
      |s_t^n) \sum_{t=0}^T 
      \gamma^t \nabla \chi_{a_t^n,s_t^n}(\theta_i) D_{j,t}^n d\Omega
    \\
    &= \nabla V_j(\theta_i),
  \end{align*}
  %
  %
  where $d\Omega$ and $\Gamma$ are defined as in the proof of
  Proposition~\ref{prop:value-function-estimates}, and for the last
  equality we have also used the Policy Gradient Theorem with baseline
  (cf.~\cite[Section 13.4]{RSS-AGB:18}).  Therefore,
  \begin{align*}
    &\mathbb{E}[\widehat{\nabla V}_j(\theta_i)] \!=\!
      \frac{\bar{N}_i}{\bar{N}_i+\tilde{N}_i}\nabla V_j(\theta_i)
      \!+\! \frac{\tilde{N}_i}{\bar{N}_i+\tilde{N}_i}\nabla
      V_j(\theta_i) \!=\! \nabla V_j(\theta_i). 
  \end{align*}

  \ref{it:gradient-variance}:
  Note that for each $l\in[d]$,
  \begin{multline}\label{eq:bound-gradient-on-policy}
    \Big\rvert \sigma_j\sum_{t=0}^T
    \gamma^t
    \nabla \chi_{a_t^{n},s_t^{n}}(\theta_i)^{(l)}
    \sum_{t^{\prime}=t}^{T} \! \big( \gamma^{t^{\prime}-t}
    R_j(s_{t^\prime}^n,a_{t^{\prime}}^n,s_{t^{\prime}+1}^n) - 
    b(s_t^n) \big) \Big\rvert
    \\
    \leq \tilde{B}\sum_{t=0}^T \gamma^t \sum_{t^\prime = t}^T \Big(
    \gamma^{t^\prime-t}B_j + \hat{B} \Big) = \psi_j. 
  \end{multline}
  By Popovicius' inequality~\cite[Corollary 1]{RB-CD:00}, this implies
  \begin{align*}
    &\VaR\Bigl[ 
      \sigma_j\sum_{t=0}^T
      \gamma^t
      \nabla \chi_{a_t^{n},s_t^{n}}(\theta_i)^{(l)}
      D_{j,t}^n 
      \Bigr] \leq \psi_j^2,
  \end{align*}
  for $n\in[\bar{N}_i]$.  On the other hand, for
  $n\in[\bar{N}_i : |\Jc_i|]$,
  \begin{multline}\label{eq:bound-gradient-off-policy-term}
    \bigg\rvert \sigma_j \prod\limits_{t=0}^T \frac{
      \pi_{\theta_i}(a_t^n |s_t^n) }{ \zeta_n(a_t^n |s_t^n) }
    \sum_{t=0}^T 
    \gamma^t
    \nabla \chi_{a_t^{n},s_t^{n}}(\theta_i)^{(l)}
    D_{j,t}^n \bigg\rvert \leq
    \\
    \frac{\tilde{B}}{\nu^{T+1}} \sum_{t=0}^T \gamma^t
    \sum_{t^\prime=t}^T \Big( \gamma^{t^\prime - t} B_j + \hat{B}
    \Big) = \bar{\psi}_j. 
  \end{multline}
  Again, by Popoviciu's inequality~\cite[Corollary 1]{RB-CD:00},
  \begin{align*}
    \VaR \bigg[ 
    \sigma_j \prod\limits_{t=0}^T \frac{ \pi_{\theta_i}(a_t^n |s_t^n)
    }{\zeta_n(a_t^n |s_t^n) } \sum_{t=0}^T 
    \gamma^t \frac{\partial}{\partial \theta^{(l)} } \chi_{a_t^n,s_t^n}(\theta_i)
    D_{j,t}^n
    \bigg] \leq \bar{\psi}_j^2,
  \end{align*}
  for $n\in[ \bar{N}_i : |\Jc_i| ]$.  Since the random variables
  \begin{align*}
    &\bigg\{ 
      \sigma_j \prod\limits_{t=0}^T \frac{ \pi_{\theta_i}(a_t^n
      |s_t^n) }{ \zeta_n(a_t^n |s_t^n) } \sum_{t=0}^T 
      \gamma^t
      \nabla \chi_{a_t^{n},s_t^{n}}(\theta_i)^{(l)}
      D_{j,t}^n
      \bigg \}_{ n\in [|\Jc_i| ] }
  \end{align*}
  are independent, it follows that 
  \begin{align*}
    \VaR[ \widehat{\nabla V_j}(\theta_i)^{(l)} ] \leq \frac{
    \bar{N}_i \psi_j^2 + \tilde{N}_i\bar{\psi}_j^2 }{ |\Jc_i|^2 },
    \quad \forall l\in[d].  
  \end{align*}
  Note also that~\eqref{eq:bound-gradient-on-policy}
  and~\eqref{eq:bound-gradient-off-policy-term} imply that
  $\widehat{\nabla V_j}(\theta_i)^{(l)}$ is uniformly upper bounded by
  $\frac{ \bar{N}_i \psi_j + \tilde{N}_i \bar{\psi}_j }{ |\Jc_i| }$.

  \ref{it:gradient-tail-probability}: From Hoeffding's inequality,
  using~\eqref{eq:bound-gradient-on-policy},~\eqref{eq:bound-gradient-off-policy-term},
  for any $\epsilon>0$ and $l\in[d]$,
  \begin{multline*}
    \mathbb{P}\Big( 
    \Big\rvert \widehat{\nabla V_j}(\theta_i)^{(l)} - \nabla
    V_j(\theta_i)^{(l)} \Big\rvert \leq \frac{\epsilon}{\sqrt{d}}  
    \Big) \geq
    \\
    1 \! - \! 2 \exp\Big\{ 
    -\frac{\epsilon^2 |\Jc_i|^2 }{ 2d(\bar{N}_i \psi_j^2 +
      2\tilde{N}_i \bar{\psi}_j^2) } 
    \Big\} .
  \end{multline*}
  Now, note that if
  $|\widehat{\nabla V_j}(\theta_i)^{(l)} - \nabla V_j(\theta_i)^{(l)} |
  \leq \frac{\epsilon}{\sqrt{d}}$ for all $l\in[d]$, then
  $\norm{ \widehat{\nabla V_j}(\theta_i) - \nabla V_j(\theta_i) } \leq
  \epsilon$, which means that
  \begin{multline*}
    \mathbb{P}\Big( \norm{ \widehat{\nabla
        V_q}(\theta_i) - \nabla V_q(\theta_i) } \leq \epsilon \Big)
    \\
    \geq  
    \mathbb{P}\Big( \bigcap_{l=1}^d \Big\{ | \widehat{\nabla
      V_q}(\theta_i)^{(l)} - \nabla V_q(\theta_i)^{(l)} | \leq
    \frac{\epsilon}{\sqrt{d}} 
    \Big\}  \Big). 
  \end{multline*}
  Using Fr\'echet's Inequality~\cite{MF:35},
  \begin{multline*}
    \mathbb{P}\Big( \bigcap_{l=1}^d \Big\{ | \widehat{\nabla
      V_q}(\theta_i)^{(l)} -  
    \nabla V_q(\theta_i)^{(l)} | \leq \frac{\epsilon}{\sqrt{d}}
    \Big\}  \Big) 
    \geq 
    \\
    1- 2 d \exp\Big\{ 
    -\frac{\epsilon^2 |\Jc_i|^2 }{ 2d
      (\bar{N}_i\psi_j^2 + \tilde{N}_i \bar{\psi}_j^2) } 
    \Big\},
  \end{multline*}
  and the result follows.

  \ref{it:gradient-variance-extra}:
  Under~\eqref{eq:extra-assumption-difference-log-policies}, for each
  $n\in[|\Jc_i|]$,
  \begin{align}\label{eq:bound-gradient-off-policy-term-extra-assumption}
    &\bigg\rvert \sigma_j \prod\limits_{t=0}^T \frac{
      \pi_{\theta_i}(a_t^n |s_t^n) }{ \pi_{\bar{\theta}_n}(a_t^n
      |s_t^n) } \sum_{t=0}^T 
      \gamma^t \frac{\partial}{\partial \theta^{(l)} } \chi_{a_t^n,s_t^n}(\theta_i)
      D_{j,t}^n \bigg\rvert \leq 
    \\ 
    &\tilde{B} \exp\Big\{ (T+1) \tilde{L}\norm{
      \theta_i-\bar{\theta}_n } \Big\} \sum_{t=0}^T \gamma^t
      \sum_{t^\prime=t}^T \Big( \gamma^{t^\prime - t} B_j + \hat{B}
      \Big) = \tilde{\psi}_{i,j,n}.      \notag
  \end{align}
  The argument is analogous to the one used
  in~\ref{it:gradient-variance}.
  Note~\eqref{eq:bound-gradient-off-policy-term-extra-assumption}
  implies $\widehat{\nabla V_j}(\theta_i)^{(l)}$ is uniformly upper
  bounded by
  $\frac{\sum_{n=1}^{|\Jc_i|} \tilde{\psi}_{i,j,n} }{ |\Jc_i| }$.
  
  \ref{it:gradient-tail-probability-extra}: This follows analogously
  to item~\ref{it:gradient-tail-probability} by
  using~\eqref{eq:bound-gradient-off-policy-term-extra-assumption}.
\end{proof}


Propositions~\ref{prop:value-function-estimates}
and~\ref{prop:gradient-value-function-estimates} characterize the
statistical properties of the estimates of the value functions and
their gradients, generalizing to the on/off-policy case our previous
result in~\cite[Lemma 2]{PM-AM-JC:25-l4dc}, which was limited to the
on-policy case.  These results show that, by increasing the number of
episodes (either on-policy or off-policy) used, the distribution of
the estimates of the value functions and their gradients concentrates
around their true values, with the rate of concentration depending on
the constants defined in
Assumptions~\ref{as:importance-sampling-ratio-well-defined},~\ref{as:boundedness-reward},~\ref{as:differentiability-lipschitzness-policy}.

\begin{remark}\longthmtitle{Assumption on global Lipschitzness} {\rm
    Assumption~\eqref{eq:extra-assumption-difference-log-policies} is
    standard in the literature (cf.~\cite[Assumption
    3.1]{KZ-AK-HZ-TB:20}). We note that, if the parameterized policy
    $\pi_{\theta}$ is globally Lipschitz uniformly in $a$ and $s$,
    then~\eqref{eq:extra-assumption-difference-log-policies} is
    satisfied.  Indeed,
    %
    using the Mean Value Theorem~\cite[Theorem 5.10]{WR:53}, and under
    Assumption~\ref{as:importance-sampling-ratio-well-defined}, we
    deduce
    \begin{align*}
      |\log\pi_{\theta}(a|s)-\log\pi_{\theta^{\prime}}(a|s)| \leq
      \frac{1}{p^*} |\pi_{\theta}(a|s)-\pi_{\theta^{\prime}}(a|s)|, 
    \end{align*}
    for some
    $p^* \in [\pi_{\theta}(a|s), \pi_{\theta^{\prime}}(a|s) ]$.  Note
    that such $p^*$ is strictly positive because of
    Assumption~\ref{as:importance-sampling-ratio-well-defined}.
    Hence, if $\pi_{\theta}$ is globally Lipschitz uniformly in $a$
    and $s$, it follows
    that~\eqref{eq:extra-assumption-difference-log-policies} is
    satisfied.
    This is the case for truncated Gaussian policies with compact
    state and action spaces (cf.~\cite[Section
    6]{KZ-AK-HZ-TB:20-arxiv},~\cite[Section 5]{PM-AM-JC:25-l4dc}).
    \demo }
\end{remark}

\subsection{Safety Guarantees}

In this section we study the safety guarantees of RSGF-RL.

\begin{theorem}\longthmtitle{Safety
    guarantees}\label{th:safety-guarantees}
  Suppose
  Assumptions~\ref{as:importance-sampling-ratio-well-defined},~\ref{as:boundedness-reward},
  and~\ref{as:differentiability-lipschitzness-policy} hold.  Let
  $i\in\mathbb{Z}_{>0}$, $\bar{N}_i$, $\tilde{N}_i$, $\phi_j$, and
  $\bar{\phi}_j$ as in
  Proposition~\ref{prop:value-function-estimates}, and $\psi_j$,
  $\bar{\psi}_j$ as in
  Proposition~\ref{prop:gradient-value-function-estimates}.
  Suppose~\eqref{eq:Rhat-definition} is feasible at
  $\theta_i\in\real^d$ and the stepsize satisfies
  \begin{align}\label{eq:hi-leq-condition}
    h_i < \min \Big\{ \frac{1}{\alpha}, \frac{\beta(\theta_i)}{L_1},
    \hdots, \frac{\beta(\theta_i)}{L_q},
    \frac{\beta(\theta_{i})}{L_{q+1}} \Big\} .
  \end{align}
  For $j\in[q]$, define
  \begin{align*}
    \hat{M}_{i,j} = \frac{-(1-\alpha h_i)\widehat{V_j}(\theta_i) +
    \frac{h_i}{2}(\beta(\theta_i) -L_j h_i )\norm{
    \hat{\Rc}_{\alpha,\beta}(\theta_i) }^2 }{ 1 + h_i \norm{
    \hat{\Rc}_{\alpha,\beta}(\theta_i) } }. 
  \end{align*}
  Then, for any
  $\delta\in(0,1)$, under~\eqref{eq:discrete-time-rsgf-estimates}
  \begin{enumerate}
  \item\label{it:safety-Vj-leq0} if $\widehat{V}_j(\theta_i) \leq 0$ and
    \begin{subequations}\label{eq:conditions-Vj}
      \begin{align}
        &\frac{ |\Jc_i|^2 }{ \bar{N}_i\phi_j^2 +
          \tilde{N}_i\bar{\phi}_j^2 } \geq -\frac{2}{\hat{M}_{i,j}^2}
          \log \frac{\delta}{2},
        \\
        &\frac{ |\Jc_i|^2 }{ \bar{N}_i\psi_j^2 +
          \tilde{N}_i\bar{\psi}_j^2 } \geq -\frac{2 d}{\hat{M}_{i,j}^2}
          \log \frac{\delta}{2d}, 
    \end{align}
    \end{subequations}
    then $\mathbb{P}( V_j(\theta_{i+1}) \leq 0) \geq 1-2\delta$;
    
  \item\label{it:safety-Vj-geq0} if $\widehat{V}_j(\theta_i) >
    0$ is such that $\hat{M}_{i,j}>0$, and~\eqref{eq:conditions-Vj} holds,
    then, $\mathbb{P}(V_j(\theta_{i+1})\leq0) \geq 1-2\delta$;
    

  \item\label{it:safety-all-j} if for each $j\in[q]$ such that
    $\widehat{V}_j(\theta_i) > 0$, it holds that $\hat{M}_{i,j}>0$,
    and~\eqref{eq:conditions-Vj} holds for all $j \in [q]$, then
    $\mathbb{P}(V_j(\theta_{i+1}) \leq 0, \; \forall j\in[q]) \geq 1-2
    q \delta$;

  \item\label{it:safety-q+1} if $V_{q+1}(\theta_{i}) \leq 0$, then
    $V_{q+1}(\theta_{i+1}) \leq 0$.
  \end{enumerate}
\end{theorem}
\begin{proof}
  \ref{it:safety-Vj-leq0}: Since $\nabla V_j$ is Lipschitz with
  Lipschitz constant $L_j$,
  cf. Lemma~\ref{lem:lipschitzness-value-functions}, we
  invoke~\cite[Lemma 1.2.3]{YN:18} to deduce
  \begin{align}\label{eq:Nesterov}
    V_j(\theta_{i+1}) \! \leq \! V_j(\theta_i) \! + \! \nabla
    V_j(\theta)^\top ( \theta_{i+1}-\theta_i ) \! + \!
    \frac{L_j}{2}\norm{ \theta_{i+1}-\theta_i }^2.
  \end{align}
  This implies, using the Cauchy-Schwartz inequality, that
  \begin{align}\label{eq:Vjp-bound-Lipschitz}
    \notag
    V_j(\theta_{i+1})
    &  \leq V_j(\theta_i)-\widehat{V}_j(\theta_i) +
      \widehat{V}_j(\theta_i) +
    \\
    \notag
    & \quad \norm{\nabla
      V_j(\theta_i)-\widehat{\nabla V}_j(\theta_i)}
      \norm{\theta_{i+1}-\theta_i} +
    \\
    &\quad \widehat{\nabla V}_j(\theta_i)^\top (\theta_{i+1}-\theta_i) + \frac{L_j}{2}\norm{ \theta_{i+1} - \theta_i }^2.
  \end{align}
  Since
  $\theta_{i+1} = \theta_i + h_i \hat{\Rc}_{\alpha,\beta}(\theta_i)$,
  and by using the constraints in~\eqref{eq:Rhat-definition},
  inequality~\eqref{eq:Vjp-bound-Lipschitz} implies
  \begin{align}\label{eq:Vjp-bound-Lipschitz-2}
    \notag
    V_j(\theta_{i+1})
    & \leq V_j(\theta_i)-\widehat{V}_j(\theta_i) +
      (1-\alpha h_i)\widehat{V}_j(\theta_i) +
    \\
    \notag
    & \quad \norm{\nabla
      V_j(\theta_i)-\widehat{\nabla V}_j(\theta_i)} h_i
      \Rc_{\alpha,\beta}(\theta_i)
    \\
    & \quad - \frac{h_i}{2}( \beta(\theta_i) - L_j
      h_i ) \norm{ \Rc_{\alpha,\beta}(\theta_i) }^2. 
  \end{align}
  %
  %
  Note that~\eqref{eq:hi-leq-condition}, together with the fact that
  $\widehat{V}_j(\theta_i) \leq 0$, implies that $\hat{M}_{i,j} > 0$.
  Now, by
  Proposition~\ref{prop:value-function-estimates}\ref{it:value-function-tail-probability},
  if
  $\frac{ |\Jc_i|^2 }{ \bar{N}_i\phi_j^2 + \tilde{N}_i\bar{\phi}_j^2 }
  \geq -\frac{2}{\hat{M}_{i,j}^2} \log \frac{\delta}{2}$, then
  $\mathbb{P}(|\widehat{V}_j(\theta_i)-V_j(\theta_i)|\leq
  \hat{M}_{i,j}) \geq 1-\delta$.  On the other hand, by
  Proposition~\ref{prop:gradient-value-function-estimates}\ref{it:gradient-tail-probability},
  if
  $\frac{ |\Jc_i|^2 }{ \bar{N}_i\psi_j^2 + \tilde{N}_i\bar{\psi}_j^2 }
  \geq -\frac{2 d}{\hat{M}_{i,j}^2} \log \frac{\delta}{2 d}$, then
  $\mathbb{P} ( \norm{ \widehat{\nabla V}_j(\theta_i) - \nabla
    V_j(\theta_i) } \leq \hat{M}_{i,j} ) \geq 1-\delta$.
  Using~\eqref{eq:Vjp-bound-Lipschitz-2} and the definition
  of~$\hat{M}_{i,j}$, we deduce that, if
  $|\widehat{V}_j(\theta_i)-V_j(\theta_i)| \leq \hat{M}_{i,j}$ and
  $\norm{ \widehat{\nabla V}_j(\theta_i) - \nabla V_j(\theta_i) } \leq
  \hat{M}_{i,j}$, then $V_j(\theta_{i+1}) \leq 0$.  Now, the result
  follows by Fréchet's inequality~\cite{MF:35}.

  \ref{it:safety-Vj-geq0}: if $\hat{M}_{i,j}>0$,
  $|\widehat{V}_j(\theta_i) - V_j(\theta_i)| \leq \hat{M}_{i,j}$, and
  $\norm{ \widehat{\nabla V}_j(\theta_i) - \nabla V_j(\theta_i) } \leq
  \hat{M}_{i,j}$, then $V_j( \theta_{i+1} ) \leq 0$, even if
  $\widehat{V_j}(\theta_i) \geq 0$.
    %
  %
  The result follows by using a similar
  argument to~\ref{it:safety-Vj-leq0}.

  \ref{it:safety-all-j}: this follows
  from~\ref{it:safety-Vj-leq0},~\ref{it:safety-Vj-geq0}, and Fréchet's
  inequality~\cite{MF:35}.

  \ref{it:safety-q+1}: this follows from
  employing~\eqref{eq:Vq+1-Rhat} in~\eqref{eq:Nesterov}, combined with
  the hypothesis that $V_{q+1}(\theta_{i}) \leq 0$.
  %
  %
\end{proof}

Theorem~\ref{th:safety-guarantees}\ref{it:safety-Vj-leq0} shows that
if the number of episodes utilized to estimate $V_j(\theta_i)$ is
sufficiently large and $\widehat{V}_j(\theta_i)\leq 0$ (i.e., we
estimate that the $j$-th safety constraint is satisfied at iteration
$i$), then the next iterate of RSGF-RL satisfies the $j$-th safety
constraint with arbitrarily high probability.  Similarly,
Theorem~\ref{th:safety-guarantees}\ref{it:safety-Vj-geq0} provides
such guarantees when $\widehat{V}_j(\theta_i) \geq 0$ (i.e., we
estimate that the $j$-th safety constraint is not satisfied at
iteration $i$).  We note that $\hat{M}_{i,j}>0$ holds when
$\widehat{V}_j(\theta_i)\leq 0$ and
$\norm{\hat{\Rc}_{\alpha,\beta}(\theta_i)}$ is nonzero (which is a
reasonable assumption if $\theta_i$ is away from a KKT point).  By
continuity, this suggests that $\hat{M}_{i,j}>0$ is also satisfied in
a neighborhood of
$\setdef{\theta\in\real^d}{ \widehat{V}_j(\theta) \leq 0}$ (again
provided that $\norm{\hat{\Rc}_{\alpha,\beta}(\theta_i)}$ is nonzero),
and it becomes increasingly more difficult to satisfy with large
$\widehat{V}_j(\theta_i)$. Intuitively, this means that safety can be
ensured in the next iteration as long as safety violations in the
current iteration are not too extreme.

\begin{remark}\longthmtitle{Feasibility}\label{rem:feasibility-Rhat}
  {\rm Since the estimates of the value functions and their gradients
    converge to their true values as the number of episodes increases,
    cf. Propositions~\ref{prop:value-function-estimates}
    and~\ref{prop:gradient-value-function-estimates}, the requirement
    in Theorem~\ref{th:safety-guarantees}
    that~\eqref{eq:Rhat-definition} is feasible at $\theta_{i}$ is
    satisfied for large enough number of episodes with high probability under
    Assumptions~\ref{as:constraint-qualifications-safe-set}
    and~\ref{as:constraint-qualifications-outside-safe-set},
    cf. Lemma~\ref{lem:feas-Lipschitzness-G}.  \demo }
\end{remark}

We state next a result that provides safety guarantees over a finite
time horizon.  Its proof follows from
Theorem~\ref{th:safety-guarantees} and Fréchet's
inequality~\cite{MF:35}. We omit it for space reasons.

\begin{corollary}\longthmtitle{Safety guarantees over a finite time
    horizon}\label{cor:safety-finite-time-horizon}
  Suppose
  Assumptions~\ref{as:importance-sampling-ratio-well-defined},~\ref{as:boundedness-reward}
  and~\ref{as:differentiability-lipschitzness-policy} hold.  Let
  $H\in\mathbb{Z}_{>0}$. If, for each $i\in[H]$, the assumptions in
  Theorem~\ref{th:safety-guarantees}\ref{it:safety-all-j} hold, then
  under~\eqref{eq:discrete-time-rsgf-estimates},
$    \mathbb{P}\Big( \bigcap\limits_{i=1}^{H+1} \{ V_j(\theta_i) \leq 0,
    \; \forall j\in[q] \}  \Big) \geq 1-2qH\delta$. 
\end{corollary}

Corollary~\ref{cor:safety-finite-time-horizon} provides conditions
under which consecutive iterates of RSGF-RL probabilistically satisfy
the constraints.  Since $\delta$ is a design parameter, this guarantee
can be ensured with arbitrarily high probability.  Smaller values of
$\delta$, however, require a larger number of episodes, as reflected
in~\eqref{eq:conditions-Vj}.

%
%

\subsection{Convergence Guarantees}

Here we provide convergence guarantees for RSGF-RL.

\begin{theorem}\longthmtitle{Almost sure
    convergence}\label{th:almost-sure-convergence}
    %
  Suppose
  Assumptions~\ref{as:regularity-functions-Vi},~\ref{as:constraint-qualifications-safe-set},~\ref{as:importance-sampling-ratio-well-defined},~\ref{as:boundedness-reward},~\ref{as:differentiability-lipschitzness-policy}
  hold.  Further suppose that:
  \begin{enumerate}
  \item\label{it:V-initialization} $V_{q+1}(\theta_0) \leq 0$;
  \item\label{it:R-Theta} for all $\theta\in\Theta\backslash\Cc$,
    Slater's condition 
    holds for~\eqref{eq:R-definition} and CRC holds
    for~\eqref{eq:R-definition} at
    $(\theta,\Rc_{\alpha,\beta}(\theta))$;
    %
    %
  \item\label{it:almost-sure-conv-Rhat-feas}
    \eqref{eq:Rhat-definition} is feasible for all
    $i\in\mathbb{Z}_{>0}$;
  \item\label{it:asymptotically-vanishing-noise-assumption}
    $\lim\limits_{i\to\infty}\norm{\Rc_{\alpha,\beta}(\theta_i) -
      \hat{\Rc}_{\alpha,\beta}(\theta_i)} = 0$ with probability one;
  \item\label{it:stepsize-sequence}
    $\lim\limits_{i\to\infty}h_i = 0$, $\sum_{i=1}^\infty h_i = \infty$;
  \end{enumerate}
  Then, under~\eqref{eq:discrete-time-rsgf-estimates}, the sequence
  $\{ \theta_i \}_{i\in\mathbb{Z}_{>0}}$ converges to the set of KKT
  points of~\eqref{eq:safe-RL-problem-parameterized-policies} in
  $\Theta$ almost surely.
\end{theorem}
\begin{proof}
  Our proof proceeds by verifying that the hypotheses required
  by~\cite[Theorem 2.3.1]{HJK-DSC:78} hold and then invoking this
  result.
  First, note that $\{ \theta_i \}_{i\in\mathbb{Z}_{>0}}$ is bounded
  with probability one.  Indeed, since $V_{q+1}(\theta_0) \leq 0$
  by~\ref{it:V-initialization},
  %
  %
  it follows from~\ref{it:almost-sure-conv-Rhat-feas} and
  Theorem~\ref{th:safety-guarantees}(iv) that
  $V_{q+1}(\theta_i) \leq 0$ for all $i\in\mathbb{Z}_{>0}$.  This
  guarantees that $\norm{\theta_i} \leq \sqrt{C}$ for all
  $i\in\mathbb{Z}_{>0}$.  Second, $\Rc_{\alpha,\beta}$ is continuous
  on~$\Theta$.  Indeed, $\Rc_{\alpha,\beta}$ is locally Lipschitz on
  $\Cc$ by Lemma~\ref{lem:feas-Lipschitzness-G} and, since for all
  $\theta\in\Theta\backslash\Cc$, Slater's condition holds
  for~\eqref{eq:R-definition} and CRC holds
  for~\eqref{eq:R-definition} at
  $(\theta,\Rc_{\alpha,\beta}(\theta))$, cf.~\ref{it:R-Theta}, a
  similar argument to the one in the proof of
  Lemma~\ref{lem:feas-Lipschitzness-G} guarantees that
  $\Rc_{\alpha,\beta}$ is locally Lipschitz on~$\Theta\backslash\Cc$.
  Third, the set of KKT points
  of~\eqref{eq:safe-RL-problem-parameterized-policies} in $\Theta$ is
  globally asymptotically stable in~$\Theta$ by an argument analogous
  to that of
  Proposition~\ref{prop:convergence}\ref{it:convergence-global}.
  Furthermore, we write the dynamics as
  \begin{align*}
    \theta_{i+1} = \theta_i + h_i \Rc_{\alpha,\beta}(\theta_i) + h_i(
    \hat{\Rc}_{\alpha,\beta}(\theta_i)-\Rc_{\alpha,\beta}(\theta_i) ) .
  \end{align*}
  Note that the noise sequence
  $\hat{\Rc}_{\alpha,\beta}(\theta_i)-\Rc_{\alpha,\beta}(\theta_i)$ is
  asymptotically vanishing with probability one,
  cf.~\ref{it:asymptotically-vanishing-noise-assumption} and the
  stepsize sequence satisfies $\lim\limits_{i\to\infty}h_i = 0$,
  $\sum_{i=1}^\infty h_i = \infty$, cf.~\ref{it:stepsize-sequence}.
  Finally, taking the sequence $\{ \xi_n \}$ in the notation
  of~\cite[Theorem 2.3.1]{HJK-DSC:78} equal to zero, we conclude that
  the sequence $\{ \theta_i \}_{i\in\mathbb{Z}_{>0}}$ converges to the
  set of KKT points in $\Theta$ almost surely.
\end{proof}




\begin{remark}\longthmtitle{Assumptions in
    Theorem~\ref{th:almost-sure-convergence}}\label{rem:asymptotically-vanishing-noise-assumption}
  {\rm Requirement~\ref{it:V-initialization} on the initial policy
    estimate
    %
    %
    and~\ref{it:R-Theta} on constraint
    qualification conditions are reasonable, given our discussion above.
    The feasibility requirement in~\ref{it:almost-sure-conv-Rhat-feas} 
    follows in the setting considered in Remark~\ref{rem:feasibility-Rhat}.
    Regarding
    requirement~\ref{it:asymptotically-vanishing-noise-assumption}, 
    we
    note that, by the same argument as in
    Lemma~\ref{lem:feas-Lipschitzness-G}, the function
    $\hat{\xi}:\real^{ d(2\tilde{q}+1) + 1 } \to \real^d$ defined as
    \begin{align}\label{eq:xi-hat}
      & \hat{\xi}( \{ A_j
      \}_{j=1}^{\tilde{q}}, \{ B_j \}_{j=0}^{\tilde{q}}, C )
      =
        \argmin{\xi \in\real^d}{\norm{\xi+B_0}^2}
      \\
      \notag
      & \qquad\qquad \text{s.t.}
        \ A_j +
        B_j^\top
        \xi +
        \frac{C}{2}\norm{\xi}^2
        \leq 0,
        \
        j\in[\tilde{q}],
    \end{align}
    is locally Lipschitz.  This means that small perturbations in
     $\{ \nabla V_j \}_{j=1}^{\tilde{q}}$ and
    $\{ V_j \}_{j=1}^{\tilde{q}}$ (like the ones obtained from using estimates
    of such quantities) result in small
    perturbations in $\Rc_{\alpha,\beta}$.
    %
    %
    In particular, this implies that, for $\bar{\epsilon}>0$, there
    exists $\bar{\delta}$ such that, if
    $\norm{ \widehat{\nabla V_j}(\theta)-\nabla V_j(\theta)} <
    \bar{\delta}$ for all $j\in[q]\cup\{0\}$ and
    $\norm{ \hat{V}_j(\theta)- V_j(\theta)} < \bar{\delta}$ for all
    $j\in[q]$, then
    $\norm{
      \hat{\Rc}_{\alpha,\beta}(\theta)-\Rc_{\alpha,\beta}(\theta) } <
    \bar{\epsilon}$.  Since the estimates of the value functions and
    their gradients become arbitrarily close to their true values if a
    sufficiently large number of episodes is used
    (cf. Propositions~\ref{prop:value-function-estimates}
    and~\ref{prop:gradient-value-function-estimates}), this means that
    the condition
    $\lim\limits_{i\to\infty}\norm{
      \Rc_{\alpha,\beta}(\theta_i)-\hat{\Rc}_{\alpha,\beta}(\theta_i)
    } = 0$ is satisfied if the number of episodes used to create the
    estimates of the value functions and their gradients increases as
    the number of iterations increases. Finally, an example of a
    stepsize sequence verifying~\ref{it:stepsize-sequence}
    is~$h_i = \frac{1}{i}$.  \demo }
\end{remark}




The following result complements the almost sure convergence
established in Theorem~\ref{th:almost-sure-convergence} by providing a
bound
on the number of iterations required to converge to a neighborhood of
a KKT point. This finite iteration convergence result is based on
ideas from~\cite[Theorem 4.3]{KZ-AK-HZ-TB:20}.

\begin{theorem}\longthmtitle{Finite iteration
    convergence}\label{thm:finite-iteration-convergence}
  Suppose
  Assumptions~\ref{as:importance-sampling-ratio-well-defined},~\ref{as:boundedness-reward},~\ref{as:differentiability-lipschitzness-policy}
  hold and that~\eqref{eq:Rhat-definition}
  is feasible at every $\{ \theta_i \}_{i\in\mathbb{Z}_{\geq 0}}$.
  Let $h_0 = \frac{1}{\alpha}$ and $h_i = \frac{1}{\alpha \sqrt{i}}$
  for $i \in \mathbb{Z}_{>0}$, and assume that for each
  $\theta\in\Theta$, Slater's condition holds
  for~\eqref{eq:Rhat-definition}.
  Let $\hat{\ell}>0$ be such that
  $\norm{\hat{\Rc}_{\alpha,\beta}(\theta)} \leq \hat{\ell}$ for all
  $\theta\in\Theta$ and $B_L > 0$ such that
  $\hat{u}_j(\theta) \leq B_L$ for all $j\in[q]$ and
  $\theta\in\Theta$. Let $\epsilon_* < \frac{1}{B_L}$.  For
  $\epsilon > 0$, define
  \begin{align*}
    \Ite = \min\setdef{ i\in\mathbb{Z}_{>0} }{
    \inf\limits_{0\leq j \leq i}  \mathbb{E} \Big[
    \norm{ \hat{\Rc}_{\alpha,\beta}(\theta_j)}^2 \Big] \leq \epsilon }. 
  \end{align*}
  Let $\epsilon > 0$ and $\bar{\sigma}>0$ such that
  $\Var(\widehat{\nabla V_j}(\theta_i)^{(l)} ) \leq \bar{\sigma}$ for
  all $i\in[\Ite]$, $j\in \{ 0 \}\cup[q]$, and $l\in[d]$,
  $\norm{\hat{\Rc}_{\alpha,\beta}(\theta_0)}^2 > \epsilon$, and
  $\epsilon >
  \frac{3}{2}\hat{\ell}\bar{\sigma}(\frac{q}{\epsilon_*}+1)$.  Then,
  there exists $\kappa>0$ such that
  \begin{align*}
    \Ite  \leq \Big( \frac{\kappa}{\epsilon -
    \frac{3}{2}\hat{\ell}\bar{\sigma}(\frac{q}{\epsilon_*}+1) } \Big)^2.
  \end{align*}
  %
  %
\end{theorem}
%
%
\begin{proof}
  First, since Slater's condition holds
  for~\eqref{eq:Rhat-definition}, $\hat{\Rc}_{\alpha,\beta}$ is
  continuous on $\Theta$~\cite[Theorem 5.3]{AVF-JK:85}, and since
  $\Theta$ is compact, $\hat{\ell}$ as in the statement exists.
  Using~\eqref{eq:Nesterov}, we deduce
  \begin{align}\label{eq:V0-thetaip1-inequality}
    \notag
    &V_j(\theta_{i+1}) \leq V_j(\theta_i) + (\nabla
      V_j(\theta_i)-\widehat{\nabla V_j}(\theta) )^\top
      (\theta_{i+1}-\theta_i) + 
    \\
    &\qquad \qquad \quad \widehat{\nabla V_j}(\theta)^\top
      (\theta_{i+1}-\theta_i) + \frac{L_j \hat{\ell}^2 h_i^2}{2},
  \end{align}
  for all $j\in \{ 0 \} \cup [q]$.  Define
  $J_{+}^i=\setdef{j\in[q]}{\widehat{V_j}(\theta_i) \geq 0}$.  Let
  $\epsilon_* > 0$ as in the statement and define
  $ V_{\epsilon_*}^i = \widehat{V_0}(\theta_i) +
  \frac{1}{\epsilon_*}\sum_{j\in J_{+}^i} \widehat{V_j}(\theta_i)$.
  Equivalently, we write
  \begin{align*}
    V_{\epsilon_*}^{i+1}
    & = 
      \widehat{V_0}(\theta_{i+1})-V_0(\theta_{i+1}) +
      \frac{1}{\epsilon_*} \!\! \sum_{j\in J_{+}^{i+1}} \!\! \big(
      \widehat{V_j}(\theta_{i+1}) - V_j(\theta_{i+1}) \big)
    \\ 
    & + V_0(\theta_{i+1}) + \frac{1}{\epsilon_*} \sum_{j\in
      J_{+}^{i+1}} V_j(\theta_{i+1})
  \end{align*}
  Using~\eqref{eq:V0-thetaip1-inequality}, we have
  \begin{align}\label{eq:Vepsilon-ineq-1}
    \notag
    &V_{\epsilon_*}^{i+1} \leq \widehat{V_0}(\theta_{i+1}) - V_0(\theta_{i+1}) \! + \!
      \frac{1}{\epsilon_*}\sum_{j\in J_{+}^{i+1}}  \big(
      \widehat{V_j}(\theta_{i+1})-V_j(\theta_{i+1}) \big)
    \\
    \notag
    &\qquad + V_0(\theta_{i}) + \frac{1}{\epsilon_*}\sum_{j\in
      J_{+}^{i+1}} V_j(\theta_{i}) 
    \\
    \notag
    &\qquad + (\nabla V_0(\theta_i) -\widehat{\nabla V_0}(\theta_i)
      )^\top (\theta_{i+1}-\theta_i) 
    \\
    \notag
    &\qquad + \frac{1}{\epsilon_*} \sum_{j\in J_{+}^{i+1}} (\nabla
      V_j(\theta_i)-\widehat{\nabla V_j}(\theta_i) )^\top
      (\theta_{i+1}-\theta_i)
    \\
    \notag
    &\qquad + \widehat{\nabla V_0}(\theta_i)^\top
      (\theta_{i+1}-\theta_i) + \frac{1}{\epsilon_*} \sum_{j\in
      J_{+}^{i+1}} \widehat{\nabla V_j}(\theta_i)^\top
      (\theta_{i+1}-\theta_i)
    \\
    &\qquad + \bigg( \frac{L_0}{2} + \frac{\sum_{j=1}^q L_j }{2
      \epsilon_*} \bigg) 
      \hat{\ell} h_i^2 .
  \end{align}
  Equivalently,
  \begin{align}\label{eq:Vepsilon-ineq-2}
    \notag
    &V_{\epsilon_*}^{i+1} \leq \widehat{V_0}(\theta_{i+1}) - V_0(\theta_{i+1}) \! + \!
      \frac{1}{\epsilon_*} \! \sum_{j\in J_{+}^{i+1}}  \big(
      \widehat{V_j}(\theta_{i+1})-V_j(\theta_{i+1}) \big)
    \\
    \notag
    &\qquad + V_0(\theta_{i})-\widehat{V_0}(\theta_i) +
      \frac{1}{\epsilon_*}\sum_{j\in J_{+}^{i+1}} \! (V_j(\theta_{i}) -
      \widehat{V_j}(\theta_i) ) 
    \\
    \notag
    &\qquad + \widehat{V_0}(\theta_i) + \frac{1}{\epsilon_*}\sum_{j\in
      J_{+}^{i+1} \cap J_{+}^i } \widehat{V_j}(\theta_i) +
      \frac{1}{\epsilon_*}\sum_{j\in J_{+}^{i+1} \backslash J_{+}^i }
      \widehat{V_j}(\theta_i)
    \\ 
    \notag
    &\qquad + (\nabla V_0(\theta_i) -\widehat{\nabla V_0}(\theta_i)
      )^\top (\theta_{i+1}-\theta_i) 
    \\
    \notag
    &\qquad + \frac{1}{\epsilon_*} \sum_{j\in J_{+}^{i+1} } (\nabla
      V_j(\theta_i)-\widehat{\nabla V_j}(\theta_i) )^\top
      (\theta_{i+1}-\theta_i)
    \\
    \notag
    &\qquad + \widehat{\nabla V_0}(\theta_i)^\top
      (\theta_{i+1}-\theta_i) + \frac{1}{\epsilon_*} \sum_{j\in
      J_{+}^{i+1}} \widehat{\nabla V_j}(\theta_i)^\top
      (\theta_{i+1}-\theta_i)
    \\
    &\qquad + \bigg( \frac{L_0}{2} + \frac{\sum_{j=1}^q L_j }{2
      \epsilon_*} \bigg) 
      \hat{\ell} h_i^2 .
  \end{align}
  Using an argument analogous to the one in the proof of
  Proposition~\ref{prop:convergence}\ref{it:convergence-C} to obtain
  equation~\eqref{eq:dV0-dt}, but now with the estimates and the
  definition~\eqref{eq:Rhat-definition} of the approximated RSGF, one
  can derive, for all $i\in\mathbb{Z}_{>0}$,
  \begin{align}\label{eq:nablaV0-R}
    \notag
    &\widehat{\nabla V}_0(\theta_i)^\top \hat{\Rc}_{\alpha,\beta}(\theta_i)
      = -\Big( 1+
      \frac{\beta(\theta_i)}{2} \sum_{j=1}^q
      \hat{u}_i(\theta_i) \Big) 
      \norm{\hat{\Rc}_{\alpha,\beta}(\theta_i)}^2
    \\
    &+ \sum_{j=1}^q \alpha \hat{u}_j(\theta_i)\widehat{V_j}(\theta_i)
      \leq -\norm{\hat{\Rc}_{\alpha,\beta}(\theta_i)}^2 + \sum_{j=1}^q
      \alpha \hat{u}_j(\theta_i)\widehat{V_j}(\theta_i), 
  \end{align}
  where $\hat{u}_j(\theta)$ denotes the Lagrange multiplier associated
  to constraint $j$ in~\eqref{eq:Rhat-definition}.
  Furthermore, from the constraints in~\eqref{eq:Rhat-definition},
  \begin{align}\label{eq:nablaVj-R}
    \widehat{\nabla V}_j(\theta_i)^\top \hat{\Rc}_{\alpha,\beta}(\theta_i)
    &\leq -\alpha \widehat{V_j}(\theta_i),
  \end{align}
  for all $j\in[q]$.
  Substituting~\eqref{eq:nablaV0-R} and~\eqref{eq:nablaVj-R}
  into~\eqref{eq:Vepsilon-ineq-2}, we get
  \begin{align}\label{eq:Vepsilon-ineq-3}
    \notag
    &V_{\epsilon_*}^{i+1} \! \leq \! \widehat{V_0}(\theta_{i+1}) \! - \!
      V_0(\theta_{i+1}) \! + \! 
      \frac{1}{\epsilon_*} \! \! \sum_{j\in J_{+}^{i+1}}  \big(
      \widehat{V_j}(\theta_{i+1})-V_j(\theta_{i+1}) \big)
    \\
    \notag
    &\qquad + V_0(\theta_{i})-\widehat{V_0}(\theta_i) +
      \frac{1}{\epsilon_*} \sum_{j\in J_{+}^{i+1}} (V_j(\theta_{i}) -
      \widehat{V_j}(\theta_i) ) 
    \\
    \notag
    &\qquad + \widehat{V_0}(\theta_i) + \frac{1}{\epsilon_*}\sum_{j\in
      J_{+}^{i+1} \cap J_{+}^i } \widehat{V_j}(\theta_i) +
      \frac{1}{\epsilon_*}\sum_{j\in J_{+}^{i+1} \backslash J_{+}^i }
      \widehat{V_j}(\theta_i)
    \\ 
    \notag
    &\qquad + (\nabla V_0(\theta_i) -\widehat{\nabla V_0}(\theta_i)
      )^\top (\theta_{i+1}-\theta_i) 
    \\
    \notag
    &\qquad + \frac{1}{\epsilon_*} 
      \sum_{j\in J_{+}^{i+1} } (\nabla
      V_j(\theta_i)-\widehat{\nabla V_j}(\theta_i) )^\top
      (\theta_{i+1}-\theta_i)
    \\
    \notag
    &\qquad -h_i \norm{\hat{\Rc}_{\alpha,\beta}(\theta_i)}^2 +
      \sum_{j=1}^q \alpha h_i
      \hat{u}_j(\theta_i)\widehat{V_j}(\theta_i)
    \\ 
    &\qquad -  \!  \frac{1}{\epsilon_*} \! \! \sum_{j\in
      J_{+}^{i+1} } \alpha h_i 
      \widehat{V_j}(\theta_i) + \bigg( \frac{L_0}{2} +
      \frac{\sum_{j=1}^q L_j }{2 
      \epsilon_*} \bigg) 
      \hat{\ell} h_i^2 .
  \end{align}
  Now note that by an argument analogous to that
  of~\cite[Lemma D.1]{AA-JC:24-tac}, 
  there exists $B_L > 0$ as in the statement.
  Since the iterates $\{ \theta_i \}_{i\in\mathbb{Z}_{>0}}$
  remain bounded in $\Theta$, we have 
  $\hat{u}_j(\theta_i) \leq B_L$ for all $j\in[q]$ and $i\in\mathbb{Z}_{>0}$.
  Since $\widehat{V_j}(\theta_i) \leq 0$ for $j\notin J_{+}^i$,
  $\sum_{j=1}^q \alpha h_i \hat{u}_j(\theta_i) \widehat{V_j}(\theta_i)
  \leq \sum_{j\in J_{+}^i} \alpha h_i \hat{u}_j(\theta_i)
  \widehat{V_j}(\theta_i)$ and using $\epsilon_* < \frac{1}{B_L}$
  it follows that
  \begin{align}\label{eq:epsilon-star-bound-2}
    \notag
    &\sum_{j=1}^q \alpha h_i \hat{u}_j(\theta_i)\widehat{V_j}(\theta_i) -
    \sum_{j\in J_{+}^{i+1}} \frac{\alpha h_i}{\epsilon_*}
      \widehat{V_j}(\theta_i) \leq
    \\
    \notag
    &\sum_{ j\in J_{+}^i } \alpha h_i \hat{u}_j(\theta_i)\widehat{V_j}(\theta_i) -
    \sum_{j\in J_{+}^{i+1}} \frac{\alpha h_i}{\epsilon_*}
      \widehat{V_j}(\theta_i) \leq
    \\
    \notag
    &\sum_{j\in J_{+}^i \backslash J_{+}^{i+1} } \alpha h_i
      \hat{u}_j(\theta_i) \widehat{V_j}(\theta_i) - \sum_{j\in
      J_{+}^{i+1}\backslash J_{+}^i } \frac{\alpha h_i
      }{\epsilon_*}\widehat{V_j}(\theta_i) \leq
    \\
    &\sum_{j\in J_{+}^i \backslash J_{+}^{i+1} } \frac{\alpha
      h_i}{\epsilon_*} \widehat{V_j}(\theta_i) - \sum_{j\in
      J_{+}^{i+1}\backslash J_{+}^i } \frac{\alpha h_i }{\epsilon_*}
      \widehat{V_j}(\theta_i). 
  \end{align}
  Using the fact that
  $\sum_{j\in J_{+}^i \backslash J_{+}^{i+1} } \frac{
    \widehat{V_j}(\theta_i) }{\epsilon_*} + \sum_{j\in J_{+}^i \cap
    J_{+}^{i+1} } \frac{ \widehat{V_j}(\theta_i) }{\epsilon_*} =
  \sum_{j\in J_{+}^i } \frac{ \widehat{V_j}(\theta_i) }{\epsilon_*}$
  along with~\eqref{eq:epsilon-star-bound-2} and $\alpha h_i < 1$, we
  get
  \begin{align}\label{eq:Vepsilon-ineq-4}
    \notag
    &V_{\epsilon_*}^{i+1} \leq \widehat{V_0}(\theta_{i+1}) \! - \!
      V_0(\theta_{i+1}) + 
      \frac{1}{\epsilon_*} \! \! \sum_{j\in J_{+}^{i+1}}  \big(
      \widehat{V_j}(\theta_{i+1})-V_j(\theta_{i+1}) \big)
    \\
    \notag
    &\qquad + V_0(\theta_{i})-\widehat{V_0}(\theta_i) +
      \frac{1}{\epsilon_*}\sum_{j\in J_{+}^{i+1}} (V_j(\theta_{i}) -
      \widehat{V_j}(\theta_i) ) 
    \\
    \notag
    &\qquad + \widehat{V_0}(\theta_i) + \frac{1}{\epsilon_*}\sum_{j\in
      J_{+}^{i} } \widehat{V_j}(\theta_i)
    \\
    \notag
    &\qquad + (\nabla
      V_0(\theta_i)-\widehat{\nabla V_0}(\theta_i) )^\top
      (\theta_{i+1}-\theta_i) +
    \\
    \notag
    &\qquad \frac{1}{\epsilon_*} \sum_{j\in J_{+}^{i+1} } (\nabla
      V_j(\theta_i)-\widehat{\nabla V_j}(\theta_i) )^\top
      (\theta_{i+1}-\theta_i)
    \\
    &\qquad -h_i \norm{\hat{\Rc}_{\alpha,\beta}(\theta_i)}^2 + \bigg(
      \frac{L_0}{2} + \frac{\sum_{j=1}^q L_j }{2 
      \epsilon_*} \bigg) 
      \hat{\ell} h_i^2.
  \end{align}
  Taking expectations on both sides of~\eqref{eq:Vepsilon-ineq-4} with
  respect to the $\sigma$-algebra generated by
  $( \{ \theta_j \}_{j\in[\Ite]}, \Ic_0, \{ \Jc_j \}_{j\in[\Ite-1]}
  )$, we get
  %
  %
  %

  %
  %
  %
  %
  \begin{align}\label{eq:Vepsilon-ineq-expectation}
    \notag
    &\mathbb{E}[ V_{\epsilon_*}^{i+1} ] \leq
      \mathbb{E}[ V_{\epsilon_*}^{i} ] - h_i \mathbb{E}[
      \norm{\hat{\Rc}_{\alpha,\beta}(\theta_i)}^2 ]
    \\ 
    \notag
    &\qquad \qquad + \mathbb{E}[ (\nabla V_0(\theta_i) - \widehat{\nabla
      V_0}(\theta_i) )^\top
      (\theta_{i+1}-\theta_i) ]
    \\
    \notag
    &\qquad \qquad + \mathbb{E}[ \frac{1}{\epsilon_*}\sum_{j\in J_{+}^{i+1} }
      ( \nabla V_j(\theta_i) - \widehat{\nabla V_j}(\theta_i) )^\top
      (\theta_{i+1}-\theta_i) ]
    \\
    &\qquad \qquad + \Big( \frac{L_0}{2} + \frac{\sum_{j=1}^q L_j }{2
      \epsilon_*} \Big) \hat{\ell} h_i^2 .
  \end{align}
  Let $V_*$ be such that $V_{\epsilon_*}^{i} \geq
  V_*$ for all
  $i\in\mathbb{Z}_{>0}$ (note that such value exists because the value
  function estimates are uniformly bounded as shown in
  Proposition~\ref{prop:value-function-estimates}).
  Define $U_i= V_{\epsilon_*}^i-V_*$ and $L_* = \frac{L_0}{2} +
  \frac{\sum_{j=1}^q L_j }{2
    \epsilon_*}$.  Summing~\eqref{eq:Vepsilon-ineq-expectation} for
  $i\in\{ 0 \}\cup[\Ite]$,
  \begin{align}\label{eq:summation-0}
    \notag
    &\sum_{i=0}^{\Ite} \mathbb{E}\bigl[
      \norm{\hat{\Rc}_{\alpha,\beta}(\theta_i)}^2
      \bigr] \leq
      \sum_{i=0}^{\Ite} \Big( \frac{ \mathbb{E}[U_i] }{h_i} -
      \frac{\mathbb{E}[U_{i+1} ]}{h_i} \Big)
    \\
    \notag
    &\quad + \sum_{i=0}^{\Ite} L_* \hat{\ell} h_i +
      \sum_{i=0}^{\Ite} \frac{ 
      \mathbb{E}[ (\theta_{i+1}-\theta_i)^\top ( \nabla
      V_0(\theta_i)-\widehat{\nabla V_0}(\theta_i) ) ] }{h_i}
    \\
    &\quad + \sum_{i=0}^{\Ite} \sum_{j\in J_{+}^{i+1} }
      \frac{\mathbb{E}[ (\nabla
      V_j(\theta_i) - \widehat{\nabla V_j}(\theta_i) )^\top
      (\theta_{i+1}-\theta_i) ]}{ \epsilon_* h_i }.  
  \end{align}
  %
  %
  Note that 
  \begin{multline*}
    \sum_{i=0}^{\Ite} \Big( \frac{ \mathbb{E}[U_i] }{h_i} -
    \frac{\mathbb{E}[U_{i+1} ]}{h_i} \Big) =
    \\
    \frac{\mathbb{E}[U_0]}{h_0} +  \sum_{i=1}^{\Ite}
    \Big(\frac{1}{h_i}-\frac{1}{h_{i-1}} \Big) \mathbb{E}[ U_i ] -
    \frac{\mathbb{E}[U_{\Ite+1}] }{ h_{\Ite} }. 
  \end{multline*}
  Since $\{ U_i
  \}_{i\in\mathbb{Z}_{>0}}$ is positive and uniformly upper bounded
  (cf. Proposition~\ref{prop:value-function-estimates}), by letting
  $B_u$ be such that $|U_i| \leq B_u$ for all
  $i\in\mathbb{Z}_{>0}$, and noting that $\frac{1}{h_i} \ge
  \frac{1}{h_{i-1}}$ for all $i\in[\Ite]$,
  \begin{align*}
    &\sum_{i=0}^{\Ite} \Big( \frac{ \mathbb{E}[U_i] }{h_i} -
      \frac{\mathbb{E}[U_{i+1} ]}{h_i} \Big) \leq  
      \frac{B_u}{h_0} +  \sum_{i=1}^{\Ite}
      \Big(\frac{1}{h_i}-\frac{1}{h_{i-1}} \Big) B_u \! = \!
      \frac{B_u}{h_{\Ite}}. 
  \end{align*}
  On the other hand, by the Cauchy-Schwartz inequality,
  \begin{multline*}
    \mathbb{E}[ (\theta_{i+1}-\theta_i)^\top ( \nabla
    V_j(\theta_i)-\widehat{\nabla V_j}(\theta_i) ) ] \leq
    \\
    \sqrt{ \mathbb{E}[ \norm{\theta_{i+1}-\theta_i}^2 ] }
    \sqrt{ \mathbb{E}[ \norm{ \nabla V_j(\theta_i)-\widehat{\nabla
          V_j}(\theta_i) }^2 ] }, 
  \end{multline*}
  for all $j\in\{ 0 \}\cup[q]$.  Moreover, since
  $\norm{ \hat{\Rc}_{\alpha,\beta}(\theta_i) } \leq \hat{\ell}$ for
  all $i\in\mathbb{Z}_{>0}$, and since
  $\max\limits_{i\in[K_{\epsilon}]}\Var(\widehat{\nabla
    V_j}(\theta_i)) \leq \bar{\sigma}$ for all $j\in\{ 0 \} \cup [q]$,
  we have from~\eqref{eq:summation-0} that
  \begin{align*}
    \frac{1}{\Ite}\sum_{i=0}^{\Ite} \mathbb{E}\bigl[
    \norm{\hat{\Rc}_{\alpha,\beta}(\theta_i)}^2
    \bigr]
    &\leq \frac{B_u}{\Ite h_{\Ite} } +
      \frac{\sum_{i=1}^{\Ite} L_* 
      \hat{\ell} h_i}{ \Ite }
    \\
    &\quad + \hat{\ell} \bar{\sigma} \big(\frac{q}{\epsilon_*}+1 \big)
      \frac{\Ite + 1}{\Ite}.  
  \end{align*}
  %
  %
  Using the fact that $\sum_{i=1}^{ \Ite } i^{-a} \leq {\Ite}^{1-a} -
  1$ (cf.~\cite[page 31]{KZ-AK-HZ-TB:20-arxiv})
  %
  %
  for any $a\in(0,1)$, substituting $h_i =
  \frac{1}{\alpha\sqrt{i}}$, 
  and using $\frac{\Ite+1}{\Ite} \leq \frac{3}{2}$ (because 
  $\norm{\hat{\Rc}_{\alpha,\beta}(\theta_0)} > \epsilon$),
  we obtain
  \begin{align}\label{eq:average-upper-bound}
    \notag
    \frac{1}{\Ite}\sum_{i=0}^{\Ite} \mathbb{E}\bigl[
    \norm{\hat{\Rc}_{\alpha,\beta}(\theta_i)}^2 \bigr] \!
    &\leq \! 
      \frac{B_u \alpha }{ \sqrt{\Ite} } +
      \frac{L_*\hat{\ell}}{\alpha}(\frac{1}{\sqrt{\Ite}}-\frac{1}{\Ite})
    \\
    & \quad + \frac{3}{2}\hat{\ell}\bar{\sigma}\big(\frac{q}{\epsilon_*}+1\big).
  \end{align}
  By definition of $\Ite$, $ \mathbb{E} \bigl[
  \norm{\hat{\Rc}_{\alpha,\beta}(\theta_i)}^2 \bigr] >
  \epsilon$, for all
  $i\in\{0\}\cup[\Ite-1]$, and therefore
  from~\eqref{eq:average-upper-bound} by taking $\kappa = B_u \alpha +
  \frac{L_* \hat{\ell} }{\alpha}$,
  \begin{align*}
    \epsilon \leq \frac{1}{\Ite}\sum_{i=0}^{\Ite-1}
    \mathbb{E}\bigl[ \norm{\hat{\Rc}_{\alpha,\beta}(\theta_i)}^2
    \bigr]  
    \leq \frac{\kappa}{\sqrt{\Ite}} +
    \frac{3}{2}\hat{\ell}\bar{\sigma}\big(\frac{q}{\epsilon_*}+1\big),
  \end{align*}
  from where the result follows.
\end{proof}


%
%


We note the result in Theorem~\ref{thm:finite-iteration-convergence}
ensures the existence of $j\in[\Ite]$ such that
$\mathbb{E}[\norm{ \hat{\Rc}_{\alpha,\beta}(\theta_j) }^2] \leq
\epsilon$, but does not imply that the convergence in expectation of
the norm of $\hat{\Rc}_{\alpha,\beta}$ is monotonic.  This is akin to
the convergence results obtained for policy gradient methods
(cf.~\cite[Theorem 4.3]{KZ-AK-HZ-TB:20}).  We also point out that the
iteration number $\Ite$ is defined in terms of
$\hat{\Rc}_{\alpha,\beta}$, instead of $\Rc_{\alpha,\beta}$.  As
justified in
Remark~\ref{rem:asymptotically-vanishing-noise-assumption}, by using a
sufficiently large number of episodes when estimating the value
functions and their gradients, $\hat{\Rc}_{\alpha,\beta}$ and
$\Rc_{\alpha,\beta}$ can be made arbitrarily close at any point with
high probability.
  %
  %
This means that if the estimates of all policies obtained for
$i\in[\Ite]$ are computed with a sufficiently large number of
episodes, Theorem~\ref{thm:finite-iteration-convergence} provides a
bound for the number of iterations needed to reach a KKT point with
high probability.

\begin{remark}\longthmtitle{Assumptions in 
    Theorem~\ref{thm:finite-iteration-convergence}}\label{rem:restrictiveness-hypothesis-Theorem-finite-iteration-convergence}
  {\rm The argument in the proof of
    Theorem~\ref{thm:finite-iteration-convergence} is valid for any
    sequence $h_i = \frac{i^{-a}}{\alpha}$ for $a\in(0,1)$, but by
    following an argument similar to that of~\cite[Theorem
    4.3]{KZ-AK-HZ-TB:20}, the optimal rate is $a=1/2$, which is the
    one adopted in the statement.  Moreover,
    Proposition~\ref{prop:gradient-value-function-estimates} provides
    a way to compute the number of episodes necessary to ensure that
    the condition
    $\Var(\widehat{\nabla V_j}(\theta_i)^{(l)}) \leq \bar{\sigma}$ is
    satisfied for all $j\in\{ 0 \}\cup[q]$, $i\in[\Ite]$ and
    $l\in[d]$.  Finally, since $\Theta$ is compact, if
    $\Rc_{\alpha,\beta}$ is locally Lipschitz on $\Theta$ (e.g., if
    for all $\theta\in\Theta$, Slater's condition holds
    for~\eqref{eq:R-definition} and CRC holds
    for~\eqref{eq:R-definition} at
    $(\theta,\Rc_{\alpha,\beta}(\theta))$, cf.
    Lemma~\ref{lem:feas-Lipschitzness-G}), then $\Rc_{\alpha,\beta}$
    is bounded in $\Theta$.  Hence, $\hat{\ell}$ exists provided that
    the value function and gradient estimates are taken so that
    $\norm{\Rc_{\alpha,\beta}-\hat{\Rc}_{\alpha,\beta}}$ is bounded.
    This holds, for example, under the asymptotically vanishing noise
    assumption discussed in
    Remark~\ref{rem:asymptotically-vanishing-noise-assumption}.
    \demo}
\end{remark}

\section{Simulations}
%

Here we test RSGF-RL in two scenarios: a robot solving a navigation
task in a 2D environment and a cart-pole system seeking to keep the
pole upright by moving the cart. We compare its performance against
other approaches\footnote{The interested reader can find
  in~\cite{PM-AM-JC:25-l4dc} a comparative analysis of on-policy
  RSGF-RL with primal-dual
  approaches~\cite{SP-MCF-LFOC-AR:23,DD-KZ-TB-MJ:20}.}  that also seek
to solve constrained Markov decision processes in an anytime fashion:
constrained policy optimization (CPO)
algorithm~\cite{JS-FW-PD-AR-OK:17} and on-policy
RSGF-RL~\cite{PM-AM-JC:25-l4dc}.

\textbf{Navigation 2D}: We test the algorithm in a 2D environment,
where a robot with single-integrator dynamics navigates to a target
point while avoiding unknown obstacles
(cf. Figure~\ref{fig:navigation2d-policy-evolution}). The obstacles'
location approximates a simplified real-world floorplan.
The state space is given by $s = (x,y) \in [0,10]^2$, representing the
position of the agent, and the action space is continuous, with
$a\in[-5,5]^2$ representing the velocity in the $x$ and $y$
directions. The target point is $s^*=(8.5,8)$.  The reward is
$R_0(s,a) = -\|s-s^*\|$ and the constraint reward is
\begin{align}\label{eq:constraint-reward}
  R_1(s,a) = 
  \begin{cases}
    \varepsilon (e^{d(s)}-1), & \text{if } s \in \Cc
    \\
    1-\varepsilon, & \text{otherwise}
  \end{cases}
\end{align}
where $\varepsilon = 0.01$, $d(s)$ is the distance between $s$ and the
closest obstacle border, and the safe set $\Cc$ is the obstacle-free
region inside $[0,10]^2$. We use the family of Gaussian policies
$\pi_{\theta}(a|s) \sim \mathcal{N}(\mu_{\theta}(s),\Sigma)$, where
$\Sigma = 0.5\textbf{I}_2$ and the mean function is defined by radial
basis functions (RBF) kernels,
\begin{equation}
  \mu_\theta(s) = \sum_{i=1}^{N_c} \tanh(\theta_i) 
  \exp\left(-\frac{\|s - c_i\|^2}{2\sigma^2}\right)
\end{equation}
Here, $\tanh$ is applied element-wise, $\{c_i\}_{i=1}^{N_c}$ are the
RBF centers, and $\{\theta_i\}_{i=1}^{N_c} \subset \real^2$ are the
training parameters.  We choose the centers to be evenly spaced points
over the state space. To make a fair comparison, all algorithms
collect the same amount of episodes per iteration and perform the same
number of iterations. Table~\ref{tab:hyperparameters_comparison}
summarizes the training setup and the hyperparameters. For RSGF-RL,
the estimators are constructed using data from both the current and
the immediately preceding policies.  To mitigate the high
variance of the estimators during the training process, we clip the
values of the importance sampling weights between $0.8$ and $1.2$.

\begin{figure*}[htb!]
  \centering
  \includegraphics[width=1\linewidth]{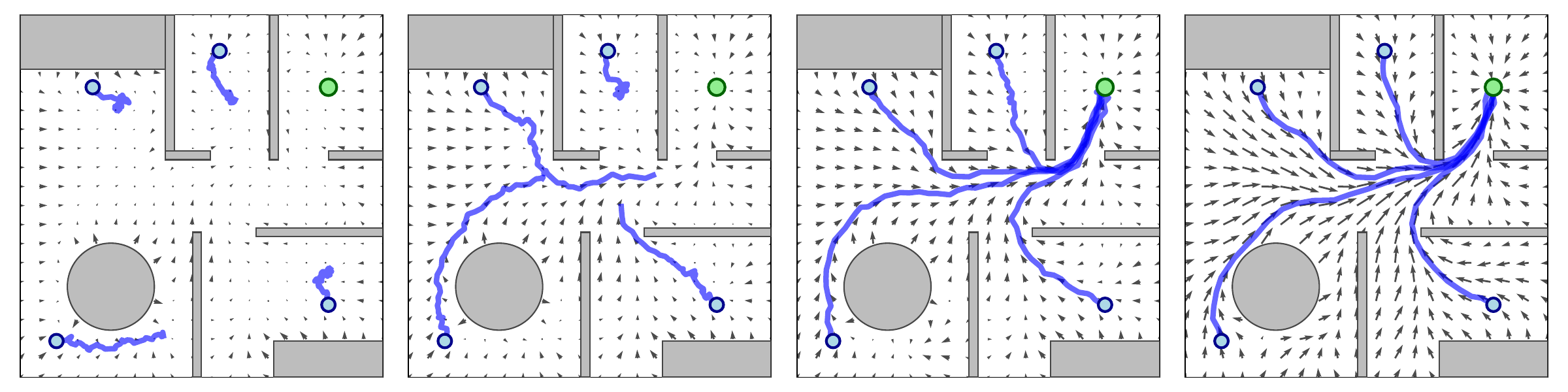}
  \caption{Policy evolution under RSGF-RL in the Navigation 2D
    example.  Obstacles are depicted in gray.  Target point in green
    and different robot initial conditions in light blue.  Initial
    policy on the left, final policy on the right, with intermediate
    policies obtained during the algorithm evolution in the middle. }
  \label{fig:navigation2d-policy-evolution}
\end{figure*}

Figure~\ref{fig:navigation2d-policy-evolution} shows the policy
evolution under the RSGF-RL algorithm starting from an initial safe
policy. The anytime nature of the algorithm is reflected in the fact
that intermediate iterations remain safe.
Figure~\ref{fig:navigation2d-results} shows the evolution of the
performance ($V_0$) and safety ($V_1$) metrics for the different
strategies.
One can see that both on- and off-policy RSGF-RL outperform CPO, while
remaining safe during the whole training procedure or recovering from
an initial unsafe state. Interestingly, off-policy RSGF-RL without
clipping the importance sampling weights performs similarly to
on-policy RSGF-RL, while RSGF-RL with clipping significantly
outperforms both in terms of sampling efficiency, converging with
fewer iterations. This suggests off-policy data can improve the
training process but introduces a high variance on the estimators that
needs to be compensated by variance reduction techniques.

\begin{figure}[htb!]
  \centering
  \includegraphics[width=1\linewidth]{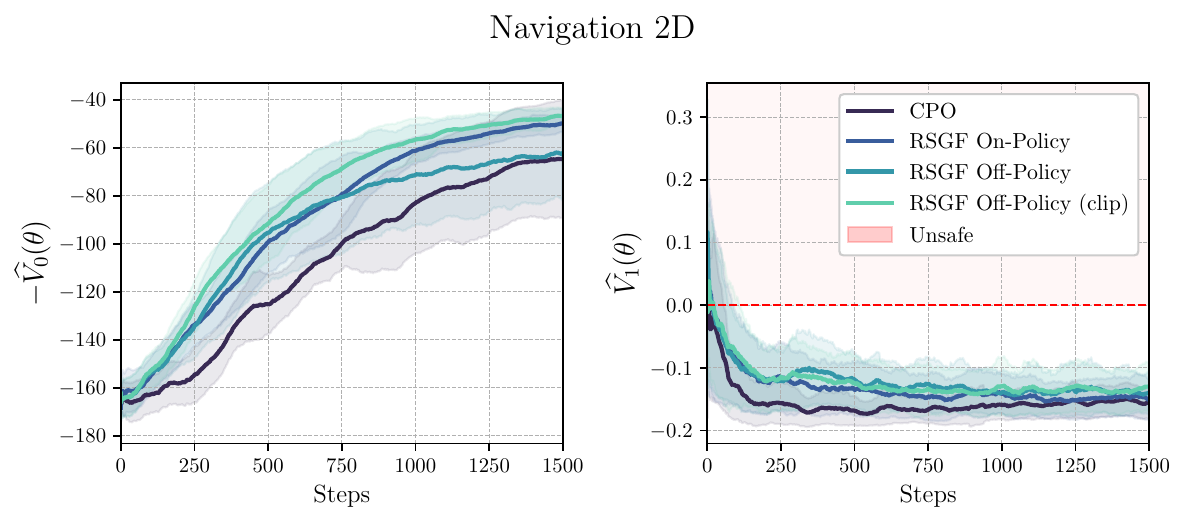}
  \caption{Comparison between CPO and different RSGF-RL training
    strategies in the Navigation 2D environment. Left plot shows the
    average $V_0( \theta)$ as a performance metric, while the right
    plot shows the average $V_1(\theta)$ as a safety metric. Averages
    are computed over 5 seeds and the shaded area represents the
    standard deviation.}
  \label{fig:navigation2d-results}
    \vspace*{-2ex}
\end{figure}

\begin{table}[htb!]
\centering
\caption{Hyperparameters for the simulation environments}
\begin{tabular}{l c c}
\toprule
\textbf{Parameter} & \textbf{Navigation 2D} & \textbf{Cart-pole}\\
\midrule
\textbf{Environment} & & \\
Time horizon ($T$) & 50 & 200\\
Discount factor ($\gamma$) & 0.98 & 0.995\\
Reward ($R_0$) & $-||s-s^*||$ & 1\\
Constraint reward ($R_1$) & Eq.~\eqref{eq:constraint-reward} & Eq.~\eqref{eq:constraint-reward}\\
$\varepsilon$ & 0.01 & 0.1\\

  \textbf{Policy} & & \\
  Type & Gaussian & Gaussian\\
  Centers evenly spaced over & $[0,10]^2$
                                            &
                                              \parbox[t]{2.5cm}{
                                              $[-3,3]\times[-\frac{\pi}{4},\frac{\pi}{4}]\times[-1,1]\times[-1.5,1.5]$}
  \\
  Number of centers ($N_c$) & 400 & 1000\\
  RBF centers variance & 0.5 & 0.5\\
  Policy variance ($\Sigma$) & $0.5\textbf{I}_2$ & $0.5\textbf{I}_4$\\

\textbf{Common elements} & & \\
Number of iterations ($k$) & 1500 & 300\\
Episodes per iteration ($N_i$) & 100 & 30\\
$V_0$ baseline ($b_0(s)$) & 0 & neural network\\
$V_1$ baseline ($b_1(s)$) & 0 & neural network\\

\textbf{CPO} & & \\
$\delta$ & 0.15 & $4\times10^{-4}$\\

\textbf{On-policy RSGF-RL} & & \\
Step size ($h$) & 0.1 & $10^{-3}$\\
$\alpha$ & 9 & 0.1\\

\textbf{Off-policy RSGF-RL} & & \\
  Step size ($h$)
                   & 0.1 &
   $\min\{10^{-3},\frac{0.02}{\|\hat{\mathcal{R}}_{\alpha,\beta}\|}\}$
  \\  
Episodes available ($|\mathcal{J}_i|$) & 200 & 15\\
Updates per iteration & 1 & 2\\
$\alpha$ & 9 & 0.1\\
$\beta$ (constant) & 1 & 1\\

\bottomrule
\end{tabular}
\label{tab:hyperparameters_comparison}
\end{table}

\textbf{Cart-pole}: We also evaluate RSGF-RL on the Gymnasium
\emph{Inverted Pendulum-v4} environment~\cite{MT-OGY:24}, where the
objective is to learn a policy that keeps a pole upright by applying
forces to a cart while avoiding hitting a wall.
The state is $s = (x, \theta, \dot{x}, \dot{\theta}) \in \real^4$,
where $x$ is the cart position, $\theta$ is the pole angle (relative
to vertical), $\dot{x}$ is the cart velocity, and $\dot{\theta}$ is
the pole angular velocity. The action space is continuous, with
$a \in [-3, 3]$ representing the force applied to the cart. The reward
is $R_0(s,a) = 1$, encouraging the pole to remain upright as long as
possible. The wall is at $x=0.5$ and hence the safe set is
$\Cc = \{ s = [x,\theta,\dot{x},\dot\theta]\in\real^4 : x < 0.5 \}$.




We define the constraint reward as in~\eqref{eq:constraint-reward},
with $\varepsilon = 0.1$ and $d(s) = [1,0,0,0]^\top s - 0.5$ and use
the same Gaussian policy with centers uniformly distributed over the
state space.  Table~\ref{tab:hyperparameters_comparison} gathers the
training details. To showcase the flexibility of the off-policy
approach, we update the policy twice per iteration using two
minibatches, instead of relying on previous trajectories (i.e., we run
steps 5-8 twice in Algorithm~\ref{alg:rl-rsgf}).

Figure~\ref{fig:InvertedPendulum-results} shows the evolution of the
performance ($V_0$) and safety ($V_1$) metrics for the different
algorithms. All approaches maintain the safety constraint below 0,
remaining safe during training. Both on-policy and off-policy RSGF-RL
outperform CPO, and RSGF-RL with clipping of the importance sampling
weights significantly outperforms all the other approaches.

\begin{figure}[htb!]
  \centering
  \includegraphics[width=1\linewidth]{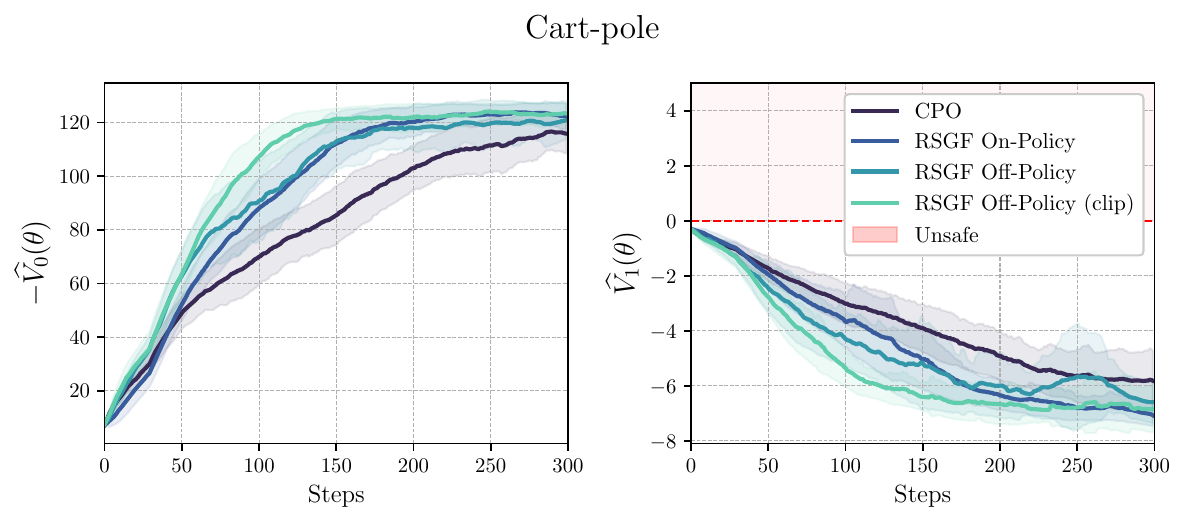}
  \caption{Comparison between CPO and different RSGF-RL training
    strategies in the Cart-pole environment. Left plot shows the
    average $V_0( \theta)$ as a performance metric, while the right
    plot shows the average $V_1(\theta)$ as a safety metric. The max
    reward is 126.61. Averages are computed over 5 seeds and the
    shaded area represents the standard deviation.}
  \label{fig:InvertedPendulum-results}
  \vspace*{-2ex}
\end{figure}


\section{Conclusions}\label{se:conclusions}

We have introduced the Robust Safe Gradient Flow-based Reinforcement
Learning (RSGF-RL) algorithm for constrained reinforcement learning
with anytime safety guarantees.  RSGF-RL's design is based on the
Robust Safe Gradient Flow, a continuous-time algorithm for anytime
constrained optimization whose forward invariance and asymptotic
stability properties we have also characterized.  At every iteration,
RSGF-RL uses off-policy episodic data to construct estimates of the
value functions defining the constrained RL problem, as well as of
their gradients.  We have rigorously characterized the statistical
properties of such estimates. Building on this, we have determined the
number of episodes needed to ensure, with a user-specified
reliability, that safe policies remain safe at the next iteration or,
alternatively, that the algorithm returns to safety from an unsafe
policy. Leveraging the theory of stochastic approximation, we have
also shown that RSGF-RL converges to a KKT point almost surely, and we
have provided a bound on the number of iterations required for
convergence. Simulations have compared the performance of RSGF-RL with
the state of the art.  Future work will focus on extensions to other
safety constraints commonly used in safe RL, such as probabilistic or
conditional value-at-risk. We also plan to explore schemes that
adaptively tune algorithm parameters
for improved convergence, safety, and memory allocation requirements.
Finally, we will extend the framework to actor-critic methods and
perform tests in physical hardware.


\bibliography{../bib/alias,../bib/JC,../bib/Main-add,../bib/Main}
\bibliographystyle{IEEEtran}

\appendix

\section{Lipschitz constants}\label{sec:appendix-Lipschitz-constants}

\begin{lemma}\longthmtitle{Lipschitzness of gradient of value
    functions}\label{lem:lipschitzness-value-functions}
  Suppose Assumptions~\ref{as:boundedness-reward}
  and~\ref{as:differentiability-lipschitzness-policy} hold. Let
  $j\in[q]\cup\{ 0 \}$.  Then, $\nabla V_j$ is Lipschitz with constant
  \begin{align*}
    &B_j L \Big( \frac{1-\gamma^T}{1-\gamma} \Big)^2 + 
      2 B_j\tilde{B}^2 \gamma \frac{ 1-(T+1)\gamma^T + T\gamma^{T+1}
      }{ (1-\gamma)^2 } +
    \\
    &B_q \tilde{B}^2 \Big( \frac{1-\gamma^T}{1-\gamma} \Big)^2.
  \end{align*}
\end{lemma}
\begin{proof}
  By the Policy Gradient Theorem~\cite[Section~13.2]{RSS-AGB:18}, for
  any $\theta\in\real^d$,
  \begin{align*}
    \nabla V_j(\theta) \! = \! \sum_{t=0}^T \! \sum_{\tau=0}^{T}
    \int\limits_{\Ic} \! \! 
    \gamma^{t+\tau} R_j(s_{t+\tau}, \! a_{t+\tau} , \! s_{t+\tau+1})
    \nabla \chi_{a_t,s_t} p_{\theta} d\sigma, 
  \end{align*}
  where $\Ic = \Sc^{T+1}\times\Ac^{T+1}$,
  $d\sigma = ds_0 ds_1 \hdots ds_T da_0 da_1 \hdots da_T$, and
  \begin{align*}
    p_{\theta} = \bigg( \prod_{k=0}^{t+\tau} P(s_{k+1},s_k,a_k) \bigg)
    \bigg( \prod_{k=0}^{t+\tau}\pi_{\theta}(a_k|s_k) \bigg) \eta(s_0). 
  \end{align*}
  Now, by following the same steps as in the proof of~\cite[Lemma
  3.2]{KZ-AK-HZ-TB:20},
  \begin{multline*}
    \norm{\nabla V_j(\theta_1)-\nabla V_j(\theta_2)}
    \leq
      \sum_{t=0}^T
      \sum_{\tau=0}^T  
      \gamma^{t+\tau} B_j L\norm{\theta_1-\theta_2} +
    \\
    \sum_{t=0}^T \sum_{\tau=0}^T \gamma^{t+\tau} B_j \tilde{B}^2
      (t+\tau+1)\norm{\theta_1-\theta_2}  .
  \end{multline*}
  Now, by using the formulas
  \begin{align*}
    &\sum_{t=0}^T \gamma^t = \frac{1-\gamma^T}{1-\gamma}, \quad 
      \sum_{t=0}^T t \gamma^t = \gamma \frac{ 1-(T+1)\gamma^T +
      T\gamma^{T+1} }{ (1-\gamma)^2 }, 
  \end{align*}
  we get
  \begin{align*}
    &\norm{\nabla V_j(\theta_1)-\nabla V_j(\theta_2)} \leq B_j L
      \norm{\theta_1-\theta_2} \Big( \frac{1-\gamma^T}{1-\gamma}
      \Big)^2
    \\
    &\qquad \qquad + 2 B_j\tilde{B}^2 \norm{\theta_1-\theta_2} \gamma
      \frac{ 1-(T+1)\gamma^T + T\gamma^{T+1} }{ (1-\gamma)^2 }
    \\
    &\qquad \qquad + B_j \tilde{B}^2 \norm{\theta_1-\theta_2} \Big(
      \frac{1-\gamma^T}{1-\gamma} \Big)^2, 
  \end{align*}
  from where the result follows.
\end{proof}

\section{Slater's Condition}\label{sec:appendix-slaters-condition}

The following result provides a sufficient condition under which
Slater's condition holds for~\eqref{eq:R-definition} for each
$\theta\in\real^d\backslash\Cc$.

\begin{lemma}\longthmtitle{Slater's condition}\label{lem:Slaters}
  Let $\delta:\real^d\to\real_{+}$ be a continuous function, and suppose
  that for each $\theta\in\real^d\backslash\Cc$, there exists
  $\xi\in\real^d$ that satisfies
  $\alpha V_j(\theta)+\nabla V_j(\theta)^\top \xi < -\delta(\theta)$.
  Consider $\xi^*:\real^d\to\real^d$ defined as
  \begin{align}\label{eq:xi-star}
    \notag
    \xi^*(\theta) &= \argmin{\xi\in\real^d}{\norm{\xi}^2} \\
                  &\text{s.t.} \ \alpha V_j(\theta)+\nabla
                    V_j(\theta)^\top \xi  \leq 0, \ j\in[q]. 
  \end{align}
  Select a differentiable function $\beta$ such that
  \begin{align*}
    \frac{\beta(\theta)}{2}\norm{\xi^*(\{ V_j(\theta), \nabla
    V_j(\theta) \}_{j=1}^q ) } < \delta(\theta).
  \end{align*}
  Then, Slater's condition holds for~\eqref{eq:R-definition} for every
  $\theta\in\real^d\backslash\Cc$.
\end{lemma}
\begin{proof}
  Since~\eqref{eq:xi-star} satisfies Slater's condition for all
  $\theta\in\real^d\backslash\Cc$, $\xi^*$ is continuous at every
  $\theta\in\real^d\backslash\Cc$~\cite[Theorem 5.3]{AVF-JK:85}.
  Therefore, a differentiable $\beta$ as required in the statement
  exists.  Now, it follows that
  $\xi^*(\{ V_j(\theta), \nabla V_j(\theta) \}_{j=1}^q )$ is strictly
  feasible for~\eqref{eq:R-definition} for each
  $\theta\in\real^d\backslash\Cc$, and the result follows.
\end{proof}

In particular, if $\delta$ is uniformly lower bounded by a positive
constant and $\xi^*$ is uniformly upper bounded, there exists a
constant $\beta$ function that makes Slater's condition hold for
$\theta\in\real^d\backslash\Cc$.  We also note that the feasibility of
the linear inequalities
$\alpha V_j(\theta) + \nabla V_j(\theta)^\top \xi < -\delta(\theta)$
can be verified using Farkas' Lemma~\cite[Theorem 22.1]{RTR:70}.
Therefore, one can verify the feasibility of such linear inequalities
and select an appropriate $\beta$ to satisfy Slater's condition in
$\real^d\backslash\Cc$.

Next, we provide a condition for CRC to hold
for~\eqref{eq:R-definition}.

\begin{lemma}\longthmtitle{Constant rank condition}\label{lem:crc}
  Let $\tilde{q} = 1$, $\theta\in\Cc$ and
  suppose~\eqref{eq:optimization-problem-V0-to-Vqtilde} satisfies
  MFCQ.  Then,~\eqref{eq:R-definition} satisfies CRC at
  $(\theta,\Rc_{\alpha,\beta}(\theta))$.
\end{lemma}
\begin{proof}
  If the single constraint of~\eqref{eq:R-definition} is not active,
  then CRC trivially holds.  Suppose that it is active.  The gradient
  with respect to $\xi$ of the single constraint
  of~\eqref{eq:R-definition} evaluated at
  $\xi = \Rc_{\alpha,\beta}(\theta)$ is
  $g_{\theta} = \nabla V_1(\theta) +
  \beta(\theta)\Rc_{\alpha,\beta}(\theta)$.  Note that if
  $g_{\theta} \neq \mathbf{0}_d$, then there exists a neighborhood
  $\Nc$ of $(\theta,\Rc_{\alpha,\beta}(\theta))$ such that if
  $(\bar{\theta},\bar{\xi})\in\Nc$, then
  $\nabla V_1(\bar{\theta}) + \beta(\bar{\theta})\bar{\xi}$ has the
  same rank as $g_{\theta}$.  Alternatively, if $g_{\theta}=0$, then
  $0 = \alpha V_1(\theta) + \nabla V_1(\theta)^\top
  \Rc_{\alpha,\beta}(\theta) +
  \frac{\beta(\theta)}{2}\norm{\Rc{\alpha,\beta}(\theta)}^2 = \alpha
  V_1(\theta) -
  \frac{\beta(\theta)}{2}\norm{\Rc_{\alpha,\beta}(\theta)}^2$.  This
  implies that $V_1(\theta) = 0$ and
  $\Rc_{\alpha,\beta}(\theta) = \mathbf{0}_d$.  Since MFCQ holds
  for~\eqref{eq:optimization-problem-V0-to-Vqtilde} and
  $V_1(\theta) = 0$, then $\nabla V_1(\theta) \neq \mathbf{0}_d$
  necessarily. However, since
  $\Rc_{\alpha,\beta}(\theta) = \mathbf{0}_d$ this contradicts the
  fact that $g_{\theta}=\mathbf{0}_d$.
\end{proof}

Finally, we state a few inequalities from probability theory 
used along the paper.

\begin{lemma}\longthmtitle{Popoviciu's inequality~\cite[Corollary
    1]{RB-CD:00}}\label{lem:popovicius-inequality}
  Let $X$ be a real-valued random variable. Let
  $m, M\in\real$ be such that $m \! \leq \! X \! \leq \! M$ almost surely.  Then,
  $\text{Var}(X) \! \leq \! \frac{(M-m)^2}{4}$.
\end{lemma}

\begin{lemma}\longthmtitle{Hoeffding's
    inequality~\cite{WH:63-hoeffding}}\label{lem:hoeffdings-inequality}
  Let $X_1, \hdots, X_n$ be independent random variables. Suppose
  there exist $a_i, b_i \in \real$ for $i\in[n]$ such that
  $a_i \leq X_i \leq b_i$ almost surely.  Let
  $S_n = X_1 + \hdots + X_n$.  Then, for any $\epsilon > 0$,
  \begin{align*}
    \mathbb{P}(|S_n - \mathbb{E}[S_n]| \geq \epsilon) \leq 2 \exp
    \Big( -\frac{2\epsilon^2}{ \sum_{i=1}^n (b_i-a_i)^2 } \Big). 
  \end{align*}
\end{lemma}

\begin{lemma}\longthmtitle{Fr\'echet's
    Inequality~\cite{MF:35}}\label{lem:Frechet-ineq}
  Let $\{ A_i \}_{i=1}^n$ be $n\in\mathbb{Z}_{>0}$ events.  Then,
  \begin{align*}
    \mathbb{P}\bigg( \bigcap_{i=1}^n A_i \bigg) \geq \max\Bigl\{ 0,
    \sum_{i=1}^n \mathbb{P}(A_i) - (n-1)  \Bigr\}. 
  \end{align*}
\end{lemma}

\vspace*{-8ex}

\begin{IEEEbiography}[{\includegraphics[width=1in,height=1.2in,clip,keepaspectratio]{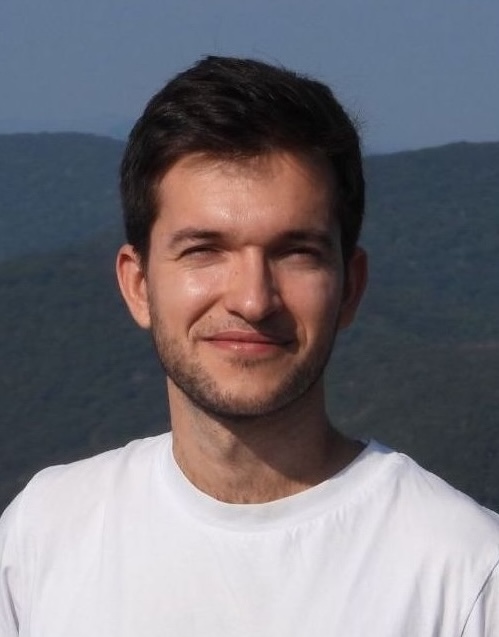}}]{Pol
    Mestres} received the Bachelor's degree in mathematics and the
  Bachelor's degree in engineering physics from the Universitat
  Polit\`{e}cnica de Catalunya, Barcelona, Spain, in 2020, and the
  Master's degree in mechanical engineering in 2021 from the
  University of California, San Diego, La Jolla, CA, USA, where he is
  currently a Ph.D candidate. His research interests include
  safety-critical control, optimization-based controllers, distributed
  optimization and motion planning.
\end{IEEEbiography}

\vspace*{-8ex}

\begin{IEEEbiography}[{\includegraphics[width=1in,height=1.2in,clip,keepaspectratio]{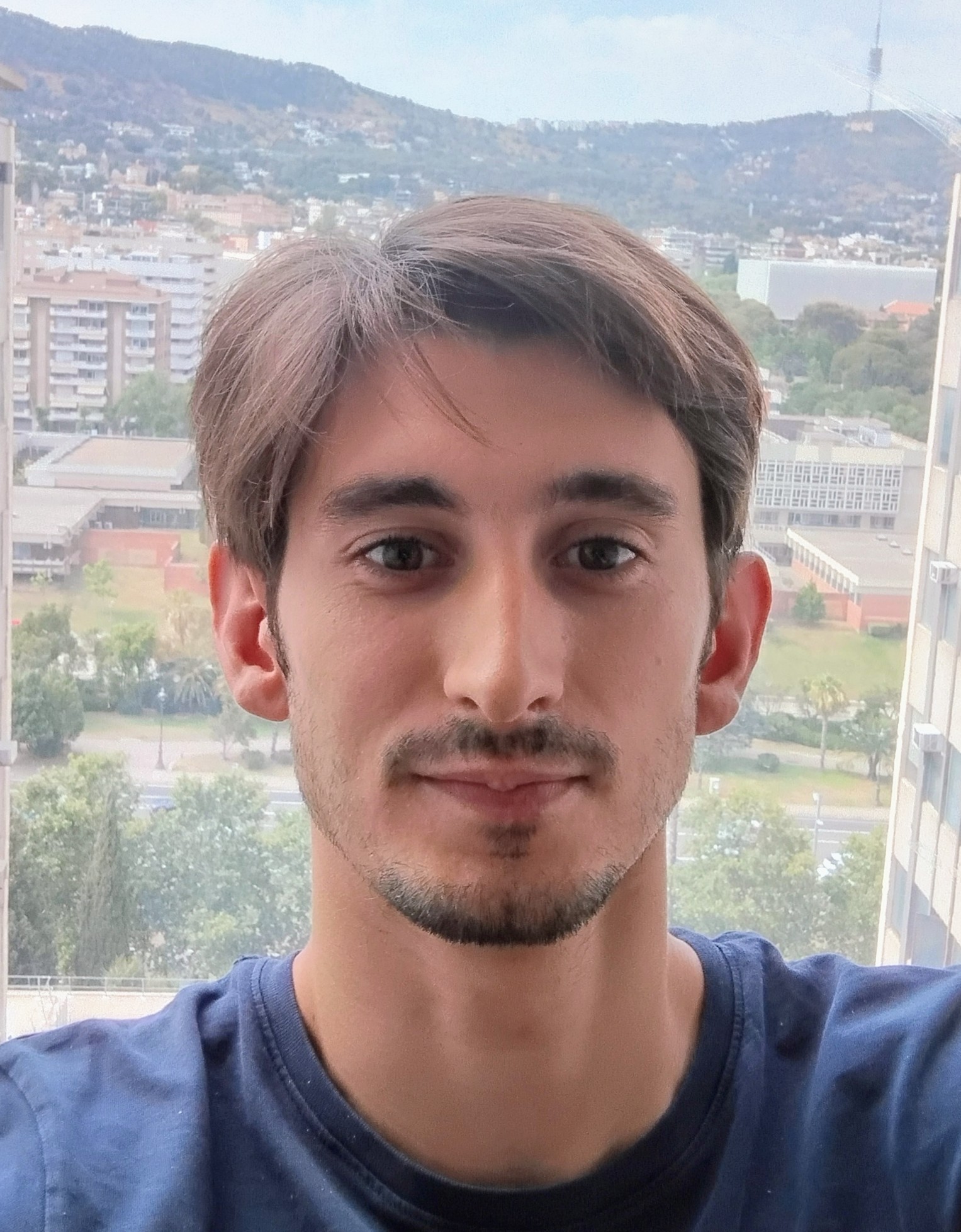}}]{Arnau
    Marzabal} received the Bachelor's degrees in Engineering Physics
  and in Industrial Engineering from the Universitat Polit\`{e}cnica
  de Catalunya, Barcelona, Spain, in 2025. He conducted his bachelor
  thesis at the University of California, San Diego, La Jolla, CA,
  USA, under the supervision of Prof. Jorge Cort{\'e}s.  His research
  interests include robotics, machine learning and
  data-driven control.
\end{IEEEbiography}

\vspace*{-8ex}

\begin{IEEEbiography}[{\includegraphics[width=1in,height=1.2in,clip,keepaspectratio]{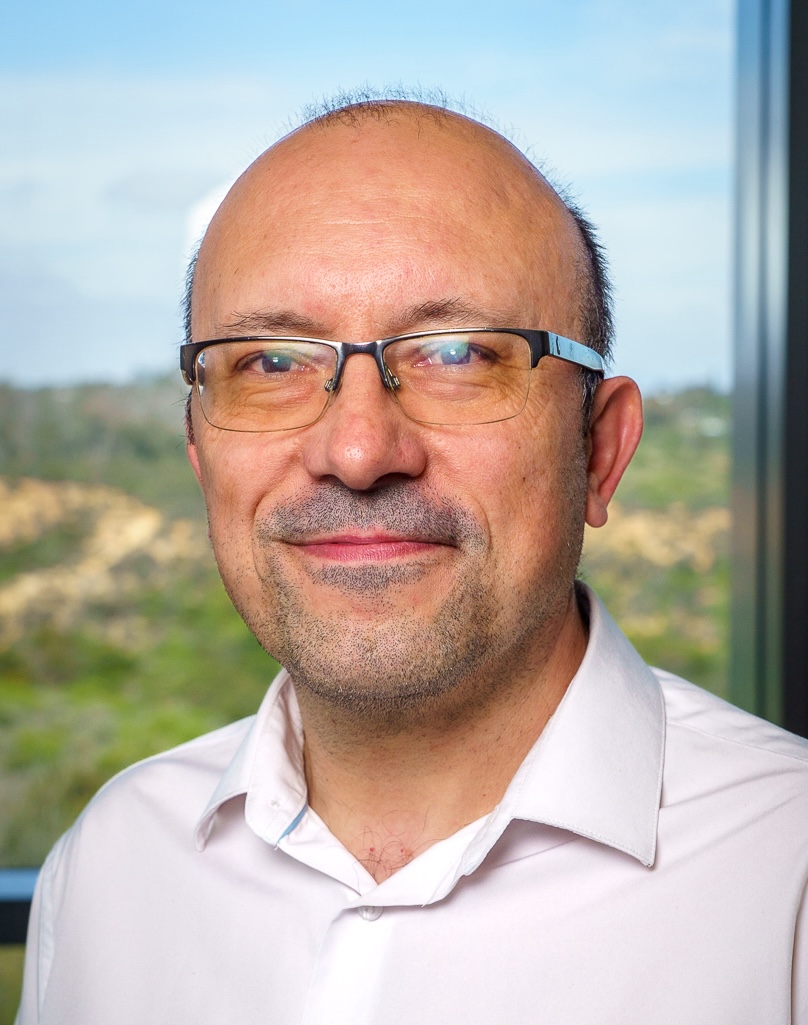}}]{Jorge
    Cort\'{e}s}(M'02, SM'06, F'14) received the Licenciatura degree in
  mathematics from Universidad de Zaragoza, Spain, in 1997, and the
  Ph.D. degree in engineering mathematics from Universidad Carlos III
  de Madrid, Spain, in 2001. He held postdoctoral positions with the
  University of Twente, Twente, The Netherlands, and the University of
  Illinois at Urbana-Champaign, Illinois, USA.
  He is a Professor and Cymer Corporation Endowed Chair in High
  Performance Dynamic Systems Modeling and Control at the Department
  of Mechanical and Aerospace Engineering, UC San Diego, California,
  USA.  He is a Fellow of IEEE, SIAM, and IFAC.  His research
  interests include distributed control and optimization, network
  science, nonsmooth analysis, reasoning and decision making under
  uncertainty, network neuroscience, and multi-agent coordination in
  robotic, power, and transportation networks.
\end{IEEEbiography}

\end{document}